\documentclass[english]{article}
\pdfoutput=1 
\usepackage[T1]{fontenc}
\usepackage[latin9]{inputenc}
\usepackage{geometry}
\usepackage{babel}
\usepackage{float}
\usepackage{units}
\usepackage{mathrsfs}
\usepackage{mathtools}
\usepackage{amsmath}
\usepackage{amsthm}
\usepackage{amssymb}
\usepackage{graphicx}
\usepackage[unicode=true,
 bookmarks=true,bookmarksnumbered=true,bookmarksopen=true,bookmarksopenlevel=1,
 breaklinks=false,pdfborder={0 0 1},backref=false,colorlinks=false]
 {hyperref}
\hypersetup{pdftitle={Improved solution to data gathering with mobile mule},
 pdfauthor={Yoad Zur and Michael Segal},
 pdfkeywords={WSN, MULE, MWCDS, Primal-Dual, UDG}}

\makeatletter

\floatstyle{ruled}
\newfloat{algorithm}{tbp}{loa}
\providecommand{\algorithmname}{Algorithm}
\floatname{algorithm}{\protect\algorithmname}

\newcommand{\lyxaddress}[1]{
	\par {\raggedright #1
	\vspace{1.4em}
	\noindent\par}
}
\theoremstyle{plain}
\newtheorem{thm}{\protect\theoremname}
\theoremstyle{definition}
\newtheorem{defn}[thm]{\protect\definitionname}
\theoremstyle{plain}
\newtheorem{prop}[thm]{\protect\propositionname}
\ifx\proof\undefined
\newenvironment{proof}[1][\protect\proofname]{\par
	\normalfont\topsep6\p@\@plus6\p@\relax
	\trivlist
	\itemindent\parindent
	\item[\hskip\labelsep\scshape #1]\ignorespaces
}{%
	\endtrivlist\@endpefalse
}
\providecommand{\proofname}{Proof}
\fi
\theoremstyle{plain}
\newtheorem{lem}[thm]{\protect\lemmaname}

\RequirePackage{fix-cm}

\usepackage{algorithmic}

\usepackage{mdframed}
\newmdenv[userdefinedwidth=1.0\columnwidth,
  innerleftmargin=0.5cm,
  rightline=false,
  bottomline=false,
  leftline=false,
  topline=false]{mdframed1}

\usepackage{fullpage}

\makeatother

\providecommand{\definitionname}{Definition}
\providecommand{\lemmaname}{Lemma}
\providecommand{\propositionname}{Proposition}
\providecommand{\theoremname}{Theorem}

\begin{document}
\title{Improved solution to data gathering with mobile mule}
\author{Yoad Zur and Michael Segal}
\date{September 17, 2019}
\maketitle

\lyxaddress{\begin{center}
Communication Systems Engineering Department,\\
School of Electrical and Computer Engineering,\\
Ben-Gurion University of the Negev, Beer-Sheva, Israel\\
Emails: yoadzu@post.bgu.ac.il, segal@bgu.ac.il
\par\end{center}}
\begin{abstract}
In this work we study the problem of collecting protected data in
ad-hoc sensor network using a mobile entity called MULE. The objective
is to increase information survivability in the network. Sensors from
all over the network, route their sensing data through a data gathering
tree, towards a particular node, called the \textit{sink}. In case
of a failed sensor, all the aggregated data from the sensor and from
its children is lost. In order to retrieve the lost data, the MULE
is required to travel among all the children of the failed sensor
and to re-collect the data. There is a cost to travel between two
points in the plane. We aim to minimize the MULE traveling cost, given
that any sensor can fail. In order to reduce the traveling cost, it
is necessary to find the optimal data gathering tree and the MULE
location. We are considering the problem for the unit disk graphs
(UDG) and Euclidean distance cost function. We propose a primal-dual
algorithm that produces a $\left(20+\varepsilon\right)$-approximate
solution for the problem, where $\varepsilon\rightarrow0$ as the
sensor network spreads over a larger area. The algorithm requires
$O\left(n^{3}\cdot\varDelta\left(G\right)\right)$ time to construct
a gathering tree and to place the MULE, where $\varDelta\left(G\right)$
is the maximum degree in the graph and $n$ is the number of nodes.
\end{abstract}

\section{Introduction}

Consider an Ad-Hoc sensor network randomly embedded in the plane.
An example application is the use of sensors scattered in forests
using aircraft, to give an indication when a fire starts. The main
problem in WSN is long-term survivability. Because of the limited
energy of the sensor, its lifetime depends on its energy consumption.
A standard solution for efficient utilization of energy is to build
a data gathering tree in the network. In this method, the data is
aggregated in order to reduce the number of messages each sensor transmits.
The sensors create a logical construction of a tree in the network,
and all the data routed to the sink according to that tree. The basic
principle of the data aggregation is that each sensor waits to aggregate
the data from all its children before it sends its own message to
its father. When a sensor failure occurs, all the aggregated data
from the sensor and its children are lost. Also, if the sensor is
not a leaf in the tree, its entire subtree is disconnected. It takes
a long time to recognize the fault, and for the network to reconstruct
a new gathering tree. Here comes the use of data MULE to retrieve
the lost data. Data MULE is a mobile unit provided with extended computing
and storage capabilities, and short-range wireless communication.

In order to gather data from the children of a failed sensor, the
MULE should travel and approach each child. After visiting all of
the children, the MULE returns to its starting point. There is a cost
to travel between two points. Given a graph of gathering sensors network,
we would like to find the data gathering tree and the position of
the MULE, which minimizes the traveling cost, while each sensor could
fail. It is common to model wireless sensor networks as a Unit Disk
Graph (UDG), where the transmission range of the sensors is defined
as one unit. The sensors are the nodes, and between each pair of sensors
that can communicate directly with each other, there will be an edge
in the graph. The MULE problem was previously defined by Crowcroft
et al. \cite{Crowcroft2016} and Yedidsion et al. \cite{Yedidsion2017}.

In this work, we focused primarily on the study of the problem on
general unit disk graphs for Euclidean distance cost function, and
only one sensor can fail at a time.

The paper is organized as follows. In the next section we present
the previous related work. In Section \ref{sec:Preliminaries-and-Model}
we define our problem and model, and introduce the main subjects that
our method is based on. In Section \ref{sec:Reduction} we show reductions
of the problem which we use later. We present our algorithm and analyze
its performance, in Section \ref{sec:The-Algorithm}. In Section \ref{sec:Empirical-Results}
we present simulation results of our algorithm, and conclude in Section
\ref{sec:Conclusions}.

\section{Related work}

In the last years, some research has been started to evaluate how
Data MULEs can benefit the performance of wireless sensor networks.
Two major approaches have developed in this area. The first approach
is the regular use of data MULEs in order to pass data between remote
components in the network \cite{Shah2003,Levin2014,Somasundara2004}.
The second approach is using data MULEs in order to increase information
survivability in the network. The main idea is to use MULEs to amend
failures in the network. This is where our work takes place. Crowcroft
et al. \cite{Crowcroft2016} use this approach to deal with the problem
of efficient data recovery using the data MULEs approach. They consider
a sensor network that uses a data gathering tree. The main idea is
to use MULEs in order to retrieve lost data caused by failed sensors
by visiting the children of the failed sensors. They define the problem
as \textquotedblleft $(\alpha,\beta)-Mule$\textquotedblright , when
$\alpha$ is the amount of simultaneously failed sensors and $\beta$
is the number of MULEs in the network. They study several aspects
of the problem and suggest a variety of solutions to the problem on
the following graph topologies: complete graph and UDG (line, random
line, and grid). Yedidsion et al. \cite{Yedidsion2017} also study
the same problem of data gathering in sensor networks using MULE.
They consider the problem when there is one MULE in the network, and
only one sensor can fail at the same time. They study the problem
for several topologies, such as UDG on a line and general UDG. Further,
they also consider a failure probability for each sensor and study
the problem for a complete graph. For the UDG on a line, they present
an $O\left(n\right)$ time algorithm that solves the problem. For
a general UDG graph, they present an algorithm and two approaches
to analyze its approximation ratio, with a tradeoff between runtime
and approximation. The first approach achieves $\left(57+\epsilon\right)$
approximation and takes $O\left(n^{3}\log n\right)$ time, where $\epsilon$
comes from a TSP calculation (according to \cite{Rao1998}) for the
approximation ratio. The second approach reduces operations and achieves
$O\left(n^{2}\log n\right)$ time. However, it pays with a worse approximation
of $\left(114+\epsilon\right)$. Ashur \cite{Ashur2018} continued
their study on a general UDG, and came up with a very large constant-factor
approximation algorithm for the problem, which requires $O\left(n\log n\right)$
runtime.

Like Yedidsion et al. \cite{Yedidsion2017} and Ashur \cite{Ashur2018},
we will also focus on the MULE problem for a general UDG graph, when
there is one MULE, and only one sensor can fail at a time. We use
a different approach to solve the problem, and propose a primal-dual
algorithm that achieves a $\left(20+\varepsilon\right)$-approximation
for the MULE problem but has a higher runtime of $O\left(n^{3}\cdot\varDelta\left(G\right)\right)$.
We note that $\varDelta\left(G\right)$ is the maximum degree in the
graph, $n$ is the number of nodes, and $\varepsilon\rightarrow0$
as the sensor network spreads over a larger area.

\section{Preliminaries and Model\label{sec:Preliminaries-and-Model}}

\subsection{Problem definition}

\global\long\def\MULE{\left(1,1\right)-MULE-UDG}%

\begin{defn}[Data gathering tree]
Data gathering tree $T=\left(V,E_{T}\right)$ is a directed spanning
tree in graph $G=\left(V,E\right)$. Consider a spanning tree in $G$
and a root node $r$, $T$ is the rooted directed version of the spanning
tree, where from each node in the tree there is a directed path to
$r$.\label{def:Data-gathering-tree}
\end{defn}
In this work, we study the MULE problem for UDG, where there is one
MULE and only one sensor failure at a time, in the network. We will
refer to this problem as ``$\MULE$'', and it is defined as follows:
Consider a network of sensors embedded in the plane with the same
transmission range. Let $G=\left(V,E\right)$ be a connected UDG that
models the WSN, when the sensors are nodes, and an edge is defined
between a pair of nodes only if they are within the transmission range
of each other. Let $T$ be a data gathering tree which spans all the
nodes in $G$. $T$ is rooted at node $r$, and sensing data propagates
from leaf nodes to the root. One MULE is located at some node on the
graph. Let $m$ be the MULE location. The MULE can travel between
a set of nodes, where the cost function of a tour $t$ (referred to
as $c\left(t\right)$) is its total \textit{Euclidean distances}.
Let $\delta\left(v,T\right)$ be the set of children that $v$ has
in tree $T$. When sensor $v_{f}$ fails, all data gathered from its
children is lost, and the MULE must travel through all of its children,
in order to retrieve the lost data. The MULE performs a circular tour
that goes through all the nodes $\delta\left(v_{f},T\right)$ and
at the end returns to the starting position $m$. Let $t\left(m,\delta\left(v_{f},T\right)\right)$
be the shortest circular tour through all the nodes in $\left\{ m\right\} \cup\delta\left(v_{f},T\right)$,
that the MULE has to take. The objective will be to minimize the overall
MULE tour, regardless of which sensor fails. Formally, the problem
is defined as follows:
\begin{defn}[$\MULE$ problem\label{def:(1,1)-MULE-UDG}]
\ Given a connected UDG graph denoted by $G=\left(V,E\right)$,
our goal is to find the MULE location $m$ and the data gathering
tree $T$, which minimizes the objective function: 
\[
\min_{m,T}\left[\sum_{v_{f}\in V}\left[c\left(t\left(m,\delta\left(v_{f},T\right)\right)\right)\right]\right]\;.
\]
\end{defn}

\subsection{The primal-dual method}

The field of combinatorial optimization problems has been heavily
influenced by the field of Linear Programming (LP) since many combinatorial
problems can be described as linear programming problems. The \textsl{primal-dual}
method was first introduced by Dantzig, Ford, and Fulkerson \cite{Dantzig1956}
in 1956, in order to find an exact solution for linear programming
problems. The principle of the \textsl{primal-dual} method can also
be useful for finding an approximate solution in polynomial time for
NP-hard optimization problems, and this method is called ``\textsl{The
Primal-Dual Method for Approximation Algorithms}''. The main idea
is to construct a feasible solution for the LP problem (referred as
the ``\textsl{Primal LP Problem}'') from scratch, using a related
LP problem (referred as the ``\textsl{Dual LP Problem}'') that guides
us during the construction of the feasible solution. In our case,
the relationship between the \textsl{primal} and the \textsl{dual}
problems is such that, the \textsl{primal} solution serves as an upper
bound, while the \textsl{dual} solution serves as a lower bound for
the optimal solution OPT. The performance of a \textsl{primal-dual}
approximation algorithm is measured by the ratio between the \textsl{primal}
and the \textsl{dual} solutions it finds. In 1995, Goemans and Williamson
\cite{Goemans1995} proposed a general technique for designing approximation
algorithms based on this method. Their technique produces 2-approximation
algorithms for a broad set of graph problems that run in $O\left(n^{2}\log n\right)$
time. Their technique had a significant impact on our research. For
more details about the \textsl{primal-dual} method, we will refer
the reader to the papers \cite{Goemans1995,Williamson2002,Goemans1997}.

\subsection[MWCDS]{Minimum-weighted connected dominating set problem}

For a later use, we want to define the term MIS and the problem of
MWCDS.

Given a graph $G=\left(V,E\right)$, a set $S\subseteq V$ will be
called an \textit{Independent Set} (IS), if for each pair of nodes
$u,v\in S$ there is no edge $\left(u,v\right)\in E$. The set will
also be called a \textit{Maximal Independent Set} (MIS), if it is
an IS and there is no node $v\in V\setminus S$ which can be added
to the set while it will still remain independent.

Consider the graph $G=\left(V,E\right)$, a set of nodes $S\subseteq V$
will be called a \textsl{Dominating Set} (DS), if each node $v\in V$
is either in the set or is a neighbor of any node in the set. The
DS with the minimum cardinality is called a \textsl{Minimum} \textsl{Dominating
Set} (MDS). A set of nodes $S\subseteq V$ will be called \textsl{Connected
Set} (CS) if its induced graph in $G$ is connected. A set of nodes
$S\subseteq V$ will be called \textsl{Connected Dominating Set} (CDS)
if $S$ is both dominating set and connected. The CDS with the minimum
cardinality is called a \textsl{Minimum} \textsl{Connected Dominating
Set} (MCDS). Consider a weight function for each node in graph $G$,
the DS with the minimum total node weights is called a \textsl{Minimum-Weighted
Dominating Set} (MWDS) and the CDS with the minimum total node weights
is called a \textsl{Minimum-Weighted Connected Dominating Set} (MWCDS).
\begin{defn}[MWCDS problem\label{def:MWCDS}]
\ Given a connected graph denoted by $G=\left(V,E\right)$ and a
node weight function $w\left(v\right)\vcentcolon V\rightarrow R^{+}$,
our goal is to find a subset of nodes $S\subseteq V$, which is CDS
in $G$ and minimizes the objective function: 
\[
\min_{S\subseteq V}\left\{ \sum_{v\in S}w\left(v\right)\right\} \;.
\]
\end{defn}
Finding MCDS is an NP-complete problem for UDG (\cite{Clark1990}).
Hence the more general case of MWCDS is also an NP-complete problem.
In 2010 Erlebach and Matus \cite{Erlebach2010} presented a polynomial-time
algorithm that achieves approximation of $\left(4+\epsilon\right)$
for MWDS, and they combine it with a known algorithm for node-weighted
Steiner trees to achieve a $7.875$-approximation for MWCDS in UDG
graphs. In addition, independently, in 2011 Zou et al. \cite{Zou2011}
also introduced a polynomial-time algorithm that achieves $\left(4+\varepsilon\right)$-approximation
to the MWDS problem on UDG. Due to the method they were based on,
there is a trade-off between the approximation performance ratio and
the runtime. So they only proved that the runtime was polynomial,
but they did not analyze it accurately.

We wanted to put more emphasis on a reasonable runtime, and therefore
in this work, we chose to use a different method than the others:
we will use the \textit{Primal-Dual} method.

\section{Reduction of The Problem\label{sec:Reduction}}

\global\long\def\R{R}%

\global\long\def\RM{R_{M}}%

\global\long\def\dist#1#2{dist\left\{  #1,#2\right\}  }%

\global\long\def\LWF#1#2{LWF\left(#1,#2\right)}%

\global\long\def\LWB#1#2{LWB\left(#1,#2\right)}%

\global\long\def\WAC#1#2{WAC\left(#1,#2\right)}%

\global\long\def\Cone{C_{1}}%

\global\long\def\NR{N_{R}}%

\global\long\def\C{C}%

\global\long\def\PD{Primal-Dual}%

\global\long\def\MWCDSw{MWCDS}%

\global\long\def\MWCDSwIP{\MWCDSw-IP}%

\global\long\def\MWCDSwLP{\MWCDSw-LP}%

\global\long\def\MWCDSwD{\MWCDSw-D}%

\global\long\def\cost#1#2{Cost\left(#1,#2\right)}%

\subsection[Reduction to MWCDS]{Reduction to an approximate MWCDS problem}

The MULE problem we study seems to be an NP-hard problem, although
it has not yet been proven. It is interesting to note that with some
reasonable assumptions, the $\MULE$ problem can be reduced to the
problem of finding a MWCDS, when the node weight is proportional to
its \textsl{Euclidean distance} from the MULE. We assume that the
MULE has a fixed transmission range of $\RM$, within the range $0<\RM<0.3\cdot\R$,
where $\R$ is the sensor's transmission range. In our case, for UDG,
$\R=1$ and the range will be $0<\RM<0.3$. In other words, the MULE
does not need to reach the exact location of the sensor, but it needs
to get close enough to communicate with it.

Let $m$ be any location of the MULE, and $T$ be some gathering tree.
The MULE's tour for a failure $t\left(m,\delta\left(v,T\right)\right)$
of node $v$ can be represented as a series of nodes according to
their order $\left\{ m,u_{1},u_{2},u_{3},\ldots,u_{\left|\delta\left(v,T\right)\right|},m\right\} $.
We will consider the tour as three separate parts. $\LWF vT$ (``Long
Walk Forward'') is the cost of the MULE's walk from node $m$ to
node $u_{1}$, for the failure of node $v$ in gathering tree $T$.
$\LWB vT$ (``Long Walk Backward'') is the cost of the MULE's walk
from node $u_{\left|\delta\left(v,T\right)\right|}$ to node $m$.
$\WAC vT$ (``Walking Among Children'') is the cost of the MULE's
walk from node $u_{1}$ to node $u_{\left|\delta\left(v,T\right)\right|}$,
through the rest of $v$'s children $u_{2},u_{3},\ldots,u_{\left|\delta\left(v,T\right)\right|-1}$
in tree $T$.

Let $\cost Tm$ be the total tours cost $\mbox{\ensuremath{\sum_{v_{f}\in V}c\left(t\left(m,\delta\left(v_{f},T\right)\right)\right)}}$
for a MULE located at $m$ and for gathering tree $T$. Let us denote
the BackBone ($BB$) of gathering tree $T$ to be the induced undirected
subgraph of $T$, which contains all the non-leaf nodes, that is,
all nodes with degree greater than one. The BackBone is denoted by
$T_{BB}=\left(V_{BB},E_{BB}\right)$, where $V_{BB}$ is the set of
the BackBone's nodes and $E_{BB}$ is its set of edges. Let $OPT_{T^{\ast}}$
be the cost of the optimal solution for the $\MULE$ problem, and
$m^{\ast}$ be the optimal MULE location, where $T^{\ast}=\left(V,E_{T}^{\ast}\right)$
is the data gathering tree of the solution. That is, $\mbox{\ensuremath{OPT_{T^{\ast}}=\cost{T^{\ast}}{m^{\ast}}}}$.
Let $V_{BB}^{\ast}$ be the BackBone nodes of $T^{\ast}$. We denote
by $\dist mv$ the \textsl{Euclidean distance} between node $v$ and
the MULE location $m$.
\begin{prop}
If we consider the MULE's transmission range as a constant in the
range $0<\RM<0.3$, then for any data gathering tree $T$ and for
each node $v\in V$, we can bound $\WAC vT\leq\Cone$, where $\Cone=\left(1+\RM\right)\left(1+\pi\left(1+\left\lceil \nicefrac{\left(1-\RM\right)}{2\RM}\right\rceil \right)\right)$.\label{prop:WAC-bounded-by-constant}
\end{prop}
\begin{proof}
Note that all the children of any node are scattered over a fixed
area due to the UDG property, a disk area of radius $1$ (i.e. the
node transmission area). One extreme possibility for the MULE's tour
will be to reach a failed node with precisely one child. In this case,
the MULE has no need to walk in the transmission area around the node,
and therefore we have $\WAC vT=0$. As the nodes density increases
in the graph, the children of a node are increasingly filling the
transmission area around it, and the MULE will have to walk a lot
more to reach all of them. Therefore, the extreme case on the other
side will be when the MULE has to walk \textbf{all over the transmission
area} of the failed node, in order to reach all of its children. Since
we assumed that the MULE has a fixed transmission range $\RM$, we
want to show that there is a fixed-length walk during which the MULE
approaches each child at a distance of $\RM$.

The description of the walk is as follows. First, the MULE reaches
the center of the area, where the failed node is located. From there,
it moves in increasing circles, until full coverage of the area (as
described in Figure \ref{fig:WAC-walk}). The first circle is with
radius $2\RM$, and for all the rest, each circle has a radius that
grows by $2\RM$ from its predecessor. Each time the MULE completes
an entire circle, it moves $2\RM$ units away in the direction of
the radius and starts the walk of the next circle. In order to cover
the whole transmission area, the amount of circles required is $\left\lceil \nicefrac{\left(\R-\RM\right)}{2\RM}\right\rceil $.

Calculating the total length of the MULE walk consists of walking
in circles $\mbox{\ensuremath{\sum_{i=1}^{\left\lceil \nicefrac{\left(\R-\RM\right)}{2\RM}\right\rceil }2\pi\cdot\left(i\cdot2\RM\right)}}$
and walking between circuits $\mbox{\ensuremath{\left\lceil \nicefrac{\left(\R-\RM\right)}{2\RM}\right\rceil \cdot2\RM}}$.
Let us define $\mbox{\ensuremath{\NR\triangleq\left\lceil \nicefrac{\left(\R-\RM\right)}{2\RM}\right\rceil }}$,
we get: 
\[
\begin{array}{rl}
\WAC vT\leq & \NR\cdot2\RM+\sum_{i=1}^{\NR}2\pi\cdot\left(i\cdot2\RM\right)\\
= & \NR\cdot2\RM+\frac{\NR}{2}\left(2\pi\cdot2\RM+2\pi\cdot\NR\cdot2\RM\right)\\
= & 2\RM\cdot\NR\left(1+\pi\left(1+\NR\right)\right)\\
= & 2\RM\cdot\left\lceil \nicefrac{\left(\R-\RM\right)}{2\RM}\right\rceil \left(1+\pi\left(1+\left\lceil \nicefrac{\left(\R-\RM\right)}{2\RM}\right\rceil \right)\right)\\
\leq & \left(1+\RM\right)\left(1+\pi\left(1+\left\lceil \nicefrac{\left(1-\RM\right)}{2\RM}\right\rceil \right)\right)\\
= & \Cone\;.
\end{array}
\]
The second inequality derives from the rounding up operation and since
$R=1$.
\end{proof}
\begin{figure}
\begin{centering}
\includegraphics[scale=0.2]{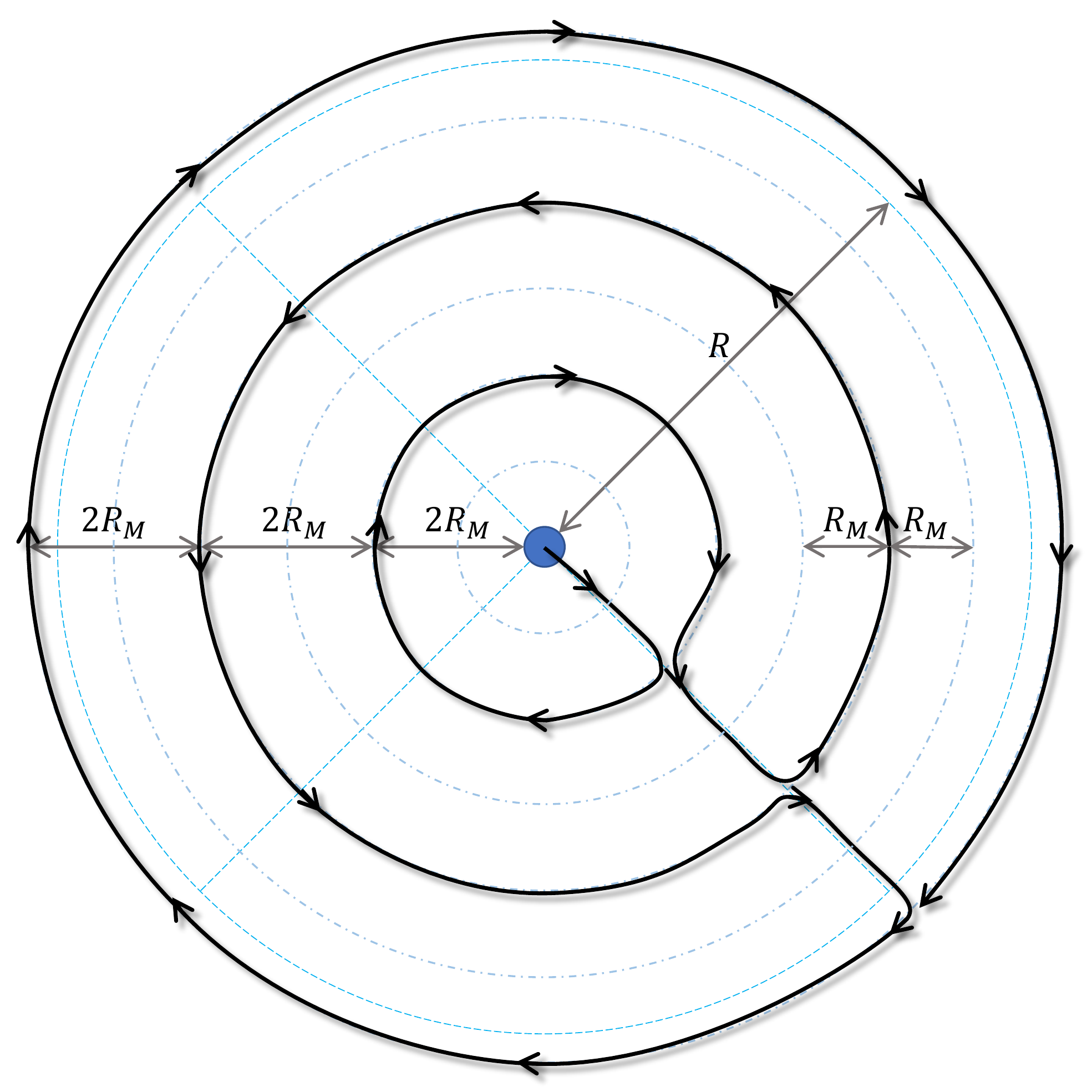}
\par\end{centering}
\caption{\label{fig:WAC-walk}Bounded walk among children of a failed sensor.}
\end{figure}

\begin{lem}
Let $\left(2\cdot\dist mv+\C\right)$ be the weight of node $v$,
where $\C$ is a positive constant dependent on $\RM$. Then $\cost Tm$
can be bounded from \textbf{above} by the total weight of $T's$ BackBone
nodes.\label{lem:OPT-bounded-from-above}
\end{lem}
\begin{proof}
Since a leaf node has no children in the tree $T$, the MULE does
not go on tour for the failure of a leaf node. Therefore, the only
nodes for which the tour cost $c\left(t\left(m,\delta\left(v_{f},T\right)\right)\right)$
is greater than $0$ are the BackBone nodes of $T$. So, we can write:
\begin{equation}
\cost Tm\triangleq\sum_{v_{f}\in V}c\left(t\left(m,\delta\left(v_{f},T\right)\right)\right)=\sum_{v\in V_{BB}}c\left(t\left(m,\delta\left(v,T\right)\right)\right)\;.\label{eq:T-R-P-2}
\end{equation}
The MULE's tour consists of $3$ different walks, $\LWF vT$, $\LWB vT$
and $\WAC vT$. Note that from the triangle inequality, and since
a child of node $v$ can be at most one unit away from $v$, we get:
$\mbox{\ensuremath{\LWF vT\leq\dist mv+1}}$ and $\mbox{\ensuremath{\LWB vT\leq\dist mv+1}}$.
Combining this with Proposition \ref{prop:WAC-bounded-by-constant},
we will get:
\begin{equation}
\begin{array}{rl}
c\left(t\left(m,\delta\left(v,T\right)\right)\right)= & \LWF vT+\WAC vT+\LWB vT\\
\leq & \left(\dist mv+1\right)+\Cone+\left(\dist mv+1\right)\\
= & 2\cdot\dist mv+2+\Cone\;.
\end{array}\label{eq:T-R-P-3}
\end{equation}
Let us define $\C\triangleq2+\Cone$. Thus, in conclusion:
\begin{equation}
\begin{array}{rl}
\cost Tm\stackrel{\left(\ref{eq:T-R-P-2}\right)}{=} & \sum_{v\in V_{BB}}c\left(t\left(m,\delta\left(v,T\right)\right)\right)\\
\stackrel{\left(\ref{eq:T-R-P-3}\right)}{\leq} & \sum_{v\in V_{BB}}\left(2\cdot\dist mv+\C\right)\;,
\end{array}\label{eq:R-T-M-1}
\end{equation}
where, $\C=3+\RM+\pi\left(1+\RM\right)\left(1+\left\lceil \nicefrac{\left(1-\RM\right)}{2\RM}\right\rceil \right)$
(Figure \ref{fig:C(Rm)}).
\end{proof}
\begin{figure}
\begin{centering}
\includegraphics[scale=0.5]{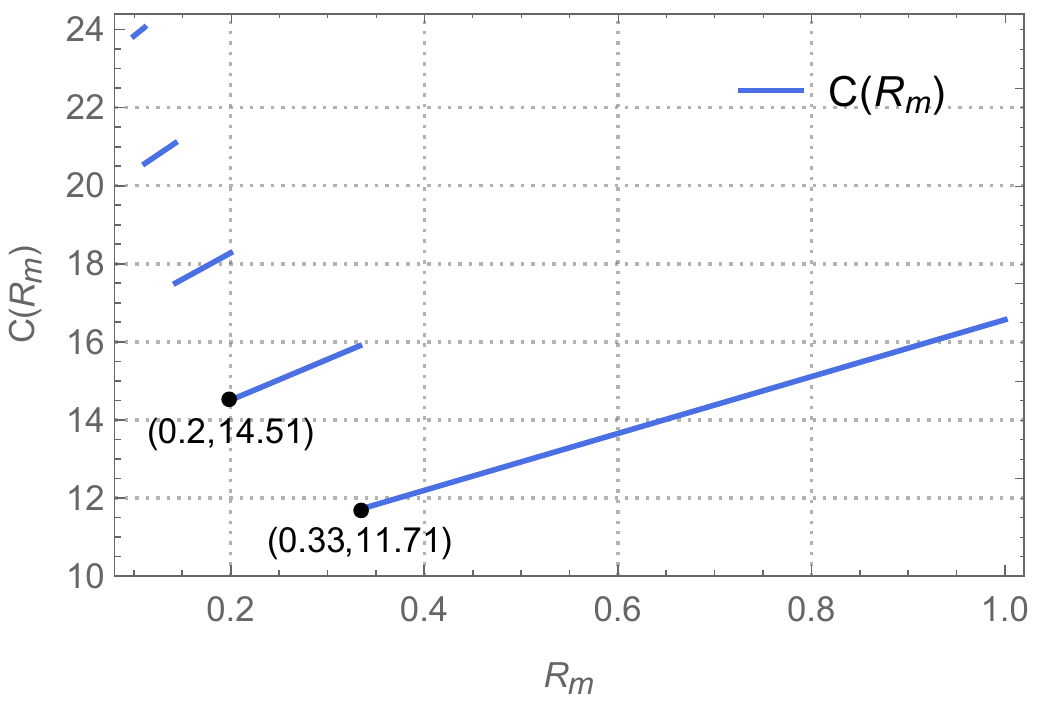}
\par\end{centering}
\caption{\label{fig:C(Rm)}Constant $C$ as a function of $\protect\RM$.}
\end{figure}

\begin{lem}
Let $\left(2\cdot\dist mv-2.6\right)$ be the weight of node $v$.
Then $\cost Tm$ can be bounded from \textbf{below} by the total weight
of $T's$ BackBone nodes.\label{lem:OPT-bounded-from-below}
\end{lem}
\begin{proof}
Note that the shortest tour to a particular failed node $v$ is in
case the node has exactly one child in $T$, which is located as close
as possible to the MULE. Since this is only one child, then $\WAC vT=0$.
We will distinguish between two possible cases: when the MULE is located
within the transmission area of the failed node $v$, and when it
is located outside it. In the first case, $\dist mv\leq1$ and the
shortest possible tour will have a cost of $0$ (in case the MULE
is located such that the only child is within its transmission range).
So we have $\mbox{\ensuremath{\LWF vT=\LWB vT=0\geq\dist mv-1.3}}$.
In the second case, where $\dist mv>1$, the location of the child
that minimizes its distance from the MULE will be directly between
the node $v$ and the MULE, at the edge of $v$'s transmission range
(exactly one unit away from $v$). This implies that $\mbox{\ensuremath{\LWF vT=\LWB vT=\dist mv-1-\RM}}$.
Recall that $\RM<0.3$. Therefore we can bound the MULE's tour from
below, as follows: 
\begin{equation}
c\left(t\left(m,\delta\left(v,T\right)\right)\right)=\LWF vT+\WAC vT+\LWB vT>2\cdot\left(\dist mv-1.3\right)\;.\label{eq:T-R-P-5}
\end{equation}

Combine it with Equation \ref{eq:T-R-P-2} of Lemma \ref{lem:OPT-bounded-from-above},
we get:
\[
\cost Tm\stackrel{\left(\ref{eq:T-R-P-2}\right)}{=}\sum_{v\in V_{BB}}c\left(t\left(m,\delta\left(v,T\right)\right)\right)\stackrel{\left(\ref{eq:T-R-P-5}\right)}{>}\sum_{v\in V_{BB}}\left(2\cdot\dist mv-2.6\right)\;.
\]
\end{proof}
\begin{thm}
Consider the optimal solution for the MWCDS problem {[}\ref{def:MWCDS}{]}
for node weights $\left(2\cdot\dist{m^{\ast}}v+\C\right)$, where
$\C$ is a positive constant dependent on $\RM$ and $m^{\ast}$ is
the optimal MULE location for the $\MULE$ problem. Let $V_{BB}^{\ast}$
be the BackBone nodes of the optimal gathering tree. If the condition
$\mbox{\ensuremath{\underset{v\in V_{BB}^{\ast}}{Average}\left[\dist{m^{\ast}}v\right]>1.3}}$
is met, then a data gathering tree whose BackBone consists of the
optimal solution to this specific MWCDS problem is an $\left(1+\varepsilon\right)$-approximate
solution for the $\mbox{\ensuremath{\MULE}}$ problem, where as the
area of the BackBone increases, the smaller the $\varepsilon$ (i.e.
$\varepsilon\rightarrow0$ for $\mbox{\ensuremath{\underset{v\in V_{BB}^{\ast}}{Average}\left[\dist{m^{\ast}}v\right]\rightarrow\infty}}$).\label{thm:Gathering-tree-approximation-by-MWCDS}
\end{thm}
\begin{proof}
We will denote $OPT_{cds^{\ast}}$ to be the total weight of the optimal
solution for the MWCDS problem, where $cds^{\ast}$ is the optimal
CDS itself, which achieves this weight. If the nodes weights are $\left(2\cdot\dist{m^{\ast}}v+\C\right)$,
then according to the problem definition {[}\ref{def:MWCDS}{]}, we
have: 
\begin{equation}
OPT_{cds^{\ast}}=\sum_{v\in cds^{\ast}}\left(2\cdot\dist{m^{\ast}}v+\C\right)\;.\label{eq:T-R-P-7}
\end{equation}

From the other side, of the $\MULE$ problem, let us define the weight
of a gathering tree $T$ (referred as $Weight_{BB}\left(T\right)$),
to be the total weight of its BackBone nodes.

We want to claim that the BackBone of $T$ is a CDS. Therefore we
are required to show that the BackBone is both dominating and connected.
$T$ is a spanning tree in $G$ by definition {[}\ref{def:Data-gathering-tree}{]}.
Thus, it is connected and contains all the nodes. Since $T$ is connected,
each node in $T$ is either a leaf or a BackBone node, and if it is
a leaf, then it has a neighbor in the BackBone. So, the BackBone is
a DS. $T$ is connected by definition {[}\ref{def:Data-gathering-tree}{]}.
Therefore, if we will prune all of its leaves it will still remain
connected. Hence, the BackBone nodes are also a CDS.

We assume that the root node is not a leaf. Otherwise, we can always
replace $T^{\ast}$ with such a tree, without affecting the approximation
ratio. Let $r$ and $v$ be the root node and one of its neighbors,
respectively. If $r$ is a leaf, we can replace $v$ to be the root
instead. Thus we will create a new tree $T'$, for which $\LWF v{T'}=\LWB v{T'}=0$
and $\WAC v{T'}\leq\WAC v{T^{\ast}}+4$. The worst case of adding
$4$ to the total tour cost of $T'$ does not affect the approximation
ratio.

The BackBone nodes of $T$ is a CDS, and the root is one of these
nodes. If the total weight of the BackBone is the weight of the tree,
then it can be said that, finding a tree with the lowest weight is
equivalent to finding a MWCDS in the graph. Formally, the following
equality holds: 
\begin{equation}
\begin{array}{rl}
\min\limits _{T}\left\{ Weight_{BB}\left(T\right)\right\} = & \min\limits _{T}\left\{ \sum\limits _{v\in V_{BB}}\left(2\cdot\dist{m^{\ast}}v+\C\right)\right\} \\
\begin{array}{c}
=\\
\\
\end{array} & \begin{array}{l}
\min\limits _{S\subseteq V}\left\{ Weight\left(S\right)\right\} \\
s.t\:S\:is\:a\:CDS
\end{array}\\
= & \sum\limits _{v\in cds^{\ast}}\left(2\cdot\dist{m^{\ast}}v+\C\right)\;.
\end{array}\label{eq:T-R-P-8}
\end{equation}

A gathering tree with the lowest weight has a lower weight than any
other tree, in particular $T^{\ast}$. Note that these are not necessarily
the same gathering trees. Therefore: 
\begin{equation}
\begin{array}{rl}
\min\limits _{T}\left\{ Weight_{BB}\left(T\right)\right\} = & \min\limits _{T}\left\{ \sum\limits _{v\in V_{BB}}\left(2\cdot\dist{m^{\ast}}v+\C\right)\right\} \\
\leq & Weight_{BB}\left(T^{\ast}\right)\\
= & \sum\limits _{v\in V_{BB}^{\ast}}\left(2\cdot\dist{m^{\ast}}v+\C\right)\;.
\end{array}\label{eq:T-R-P-9}
\end{equation}

From Lemma \ref{lem:OPT-bounded-from-above} and Lemma \ref{lem:OPT-bounded-from-below}
we have seen that $OPT_{T^{\ast}}$ can be bounded as follows: 
\[
\sum_{v\in V_{BB}^{\ast}}\left(2\cdot\dist{m^{\ast}}v-2.6\right)<OPT_{T^{\ast}}\leq\sum_{v\in V_{BB}^{\ast}}\left(2\cdot\dist{m^{\ast}}v+\C\right)\;.
\]

Now, we will calculate how much the upper bound is greater than the
optimal solution. 
\begin{equation}
\begin{array}{rl}
\frac{\sum_{v\in V_{BB}^{\ast}}\left(2\cdot\dist{m^{\ast}}v+\C\right)}{OPT_{T^{\ast}}}< & \frac{\sum_{v\in V_{BB}^{\ast}}\left(2\cdot\dist{m^{\ast}}v+\C\right)}{\sum_{v\in V_{BB}^{\ast}}\left(2\cdot\dist{m^{\ast}}v-2.6\right)}\\
= & \frac{\sum_{v\in V_{BB}^{\ast}}2\cdot\dist{m^{\ast}}v+\left|V_{BB}^{\ast}\right|\C}{\sum_{v\in V_{BB}^{\ast}}2\cdot\dist{m^{\ast}}v-\left|V_{BB}^{\ast}\right|2.6}\\
= & 1+\frac{1}{\frac{2}{\left(\C+2.6\right)}\cdot\left(\frac{\sum_{v\in V_{BB}^{\ast}}\dist{m^{\ast}}v}{\left|V_{BB}^{\ast}\right|}-1.3\right)}\\
= & 1+\frac{1}{\frac{2}{\left(\C+2.6\right)}\cdot\left(\underset{v\in V_{BB}^{\ast}}{Average}\left[\dist{m^{\ast}}v\right]-1.3\right)}\\
= & 1+\varepsilon\;,
\end{array}\label{eq:T-R-P-10}
\end{equation}
where, $\varepsilon\triangleq\left(\frac{2}{\left(\C+2.6\right)}\cdot\left(\underset{v\in V_{BB}^{\ast}}{Average}\left[\dist{m^{\ast}}v\right]-1.3\right)\right)^{-1}$,
and hence $\varepsilon\rightarrow0$ for $\mbox{\ensuremath{\underset{v\in V_{BB}^{\ast}}{Average}\left[\dist{m^{\ast}}v\right]\rightarrow\infty}}$.
That is, the $\varepsilon$ error decreases as the average distance
of the BackBone nodes from the MULE, increases. Also, as the graph
will be deployed over a wider area. We should note that for correct
approximation, the condition $\mbox{\ensuremath{\underset{v\in V_{BB}^{\ast}}{Average}\left[\dist{m^{\ast}}v\right]>1.3}}$
must be met.

Finally, we will summarize everything in one equation: 
\begin{equation}
\begin{array}{rl}
OPT_{cds^{\ast}}\stackrel{\left(\ref{eq:T-R-P-7}\right)}{=} & \sum_{v\in cds^{\ast}}\left(2\cdot\dist{m^{\ast}}v+\C\right)\\
\stackrel{\left(\ref{eq:T-R-P-8}\right)}{=} & \min\limits _{T}\left\{ \sum_{v\in V_{BB}}\left(2\cdot\dist{m^{\ast}}v+\C\right)\right\} \\
\stackrel{\left(\ref{eq:T-R-P-9}\right)}{\leq} & \sum_{v\in V_{BB}^{\ast}}\left(2\cdot\dist{m^{\ast}}v+\C\right)\\
\stackrel{\left(\ref{eq:T-R-P-10}\right)}{<} & \left(1+\varepsilon\right)OPT_{T^{\ast}}\;.
\end{array}\label{eq:R-T-M-2}
\end{equation}
\end{proof}

\subsection[Primal-dual for MWCDS]{The primal-dual method for MWCDS problem\label{subsec:Primal-dual-for-MWCDS}}

In this section, we will see how the \textsl{$\PD$} method can assist
us in finding an approximate solution to the MWCDS problem with which
we are dealing.

Consider the following Integer Program (IP) problem:
\[
\begin{array}{rlc}
{\displaystyle \min_{\bar{x}}} & {\displaystyle \sum_{v\in V}w\left(v\right)\cdot x_{v}}\\
s.t. & {\displaystyle \sum_{v\in N\left(S\right)}x_{v}\geq f\left(S\right)} & \forall S\vcentcolon S\subseteq V,S\neq\emptyset\\
 & x_{v}\in\left\{ 0,1\right\} 
\end{array}\qquad\left(\MWCDSwIP\right)\;.
\]
Let $\bar{x}$ be a vector of size $\left|V\right|$, where $x_{v}$
are its components, and let $w\left(v\right)$ be some weight function
$\mbox{\ensuremath{w\left(\cdot\right)\colon V\rightarrow R^{+}}}$.
We denote by $N\left(S\right)$ the set of nodes which neighbors to
set $S\subseteq V$. Formally: $\mbox{\ensuremath{N\left(S\right)\triangleq\left\{ u\colon\exists\,e_{u,v}\in E\colon v\in S,u\in V\setminus S\right\} }}$.
$f\left(S\right)$ is an indicator function which returns $0$ if
set $S\subseteq V$ is a DS. Formally: 
\begin{equation}
f\left(S\right)\triangleq\begin{cases}
0 & for\;N\left(S\right)\cup S=V\\
1 & otherwise
\end{cases}\;\begin{array}{c}
\,\\
.
\end{array}\label{eq:f}
\end{equation}
This IP optimization problem, denoted by $\MWCDSwIP$, is to find
a feasible vector $\bar{x}$ which minimizes the sum $W_{\bar{x}}\triangleq\sum_{v\in V}w\left(v\right)\cdot x_{v}$.

We will denote by $\MWCDSwLP$ the Linear Program (LP) relaxation
of the $\MWCDSwIP$ problem, where we relax the integrality constraint
$x_{v}\in\left\{ 0,1\right\} $ to be $x_{v}\geq0$. This is our \textsl{primal}
problem.

The \textsl{dual} problem for $\MWCDSwLP$ will be referred to as
$\MWCDSwD$, and is defined as follows: 
\[
\begin{array}{rlc}
{\displaystyle \max_{\bar{y}}} & {\displaystyle \sum_{S\subseteq V\vcentcolon S\neq\emptyset}f\left(S\right)\cdot y_{S}}\\
s.t. & {\displaystyle \sum_{S\subseteq V\vcentcolon S\neq\emptyset,v\in N\left(S\right)}y_{S}\leq w\left(v\right)} & \forall v\in V\\
 & y_{S}\geq0
\end{array}\qquad\left(\MWCDSwD\right)\;.
\]
$y_{S}$ is the components of vector $\bar{y}$. The $\MWCDSwD$ problem
is to find a feasible vector $\bar{y}$ which maximizes the sum $F_{\bar{y}}\triangleq\sum_{S\subseteq V\vcentcolon S\neq\emptyset}f\left(S\right)\cdot y_{S}$.
For more details about the $\PD$ relaxation, we will refer the reader
to \cite{Goemans1997}.

Let $\bar{x}$ be some feasible solution to $\MWCDSwIP$, and $\bar{x}^{\ast}$
is the optimal solution. Consider that $\bar{x}_{LP}^{\ast}$ is the
optimal solution to $\MWCDSwLP$. Also, let $\bar{y}$ be some feasible
dual solution to $\MWCDSwD$, and $\bar{y}^{\ast}$ is the optimal
solution. From the weak duality theorem we can state that 
\begin{equation}
W_{\bar{x}}\geq W_{\bar{x}^{\ast}}\geq W_{\bar{x}_{LP}^{\ast}}=F_{\bar{y}^{\ast}}\geq F_{\bar{y}}\;.\label{eq:T-R-P-11}
\end{equation}

Now, we argue that MWCDS and $\MWCDSwIP$ are equivalent problems.
Let $\bar{\sigma}$ be the equivalent solution of $cds^{\ast}$ to
$\MWCDSwIP$, i.e. a binary vector with ones for each $v\in cds^{\ast}$.
Symmetrically, let $\chi$ be the equivalent solution of $\bar{x}^{\ast}$
to MWCDS problem, i.e. the set of nodes with value $1$ in $\bar{x}^{\ast}$.
Note that $\mbox{\ensuremath{OPT_{cds^{\ast}}\triangleq Weight\left(cds^{\ast}\right)=W_{\bar{\sigma}}}}$
and $Weight\left(\chi\right)=W_{\bar{x}^{\ast}}$. We denote by $\text{\ensuremath{\mathscr{S}}}$
the collection of all subsets $\mbox{\ensuremath{\text{\ensuremath{\mathscr{S}}}\triangleq\left\{ S\colon S\subseteq V,S\neq\emptyset\right\} }}$.
We will show that $\bar{\sigma}$ is an optimal solution for $MWCDS-IP$.
\begin{lem}
$\bar{\sigma}$ is a feasible solution to the $\MWCDSwIP$ problem.
\label{lem:MWCDSw-is-a-MWCDSw-IP}
\end{lem}
\begin{proof}
We have to show that $\bar{\sigma}$ complies with the constraints
$\mbox{\ensuremath{\sum_{v\in N\left(S\right)}\sigma_{v}\geq f\left(S\right)\quad\forall S\in\text{\ensuremath{\mathscr{S}}}}}$
of the problem. In order to prove so, we partition the collection
$\text{\ensuremath{\mathscr{S}}}$ into three possible cases. First,
$S$ contains $cds^{\ast}$; Second, $S$ does not contain $cds^{\ast}$
and $f\left(S\right)=0$; And third, $S$ does not contain $cds^{\ast}$
and $f\left(S\right)=1$. If $S$ contains a DS, then also $f\left(S\right)=0$.
Hence, the first two cases are trivial. As for the third case, we
claim that there is at least one node $\mbox{\ensuremath{v\in N\left(S\right)}}$
which $\sigma_{v}=1$. Recall that $\sigma_{v}=1$ if and only if
$\mbox{\ensuremath{v\in cds^{\ast}}}$. If $S$ and $cds^{\ast}$
are disjoint sets, then $cds^{\ast}$ dominates all $S$ nodes. Therefore,
each node $u\in S$ has a neighbor $v\in cds^{\ast}$ (i.e. $\mbox{\ensuremath{v\in N\left(S\right)}}$).
If $S$ and $cds^{\ast}$ are not disjoint sets, then at least one
node $u\in S\cap cds^{\ast}$ must have a neighbor $v\in cds^{\ast}\setminus S$
(i.e. $\mbox{\ensuremath{v\in N\left(S\right)}}$), since $cds^{\ast}$
is connected. To conclude, $\bar{\sigma}$ is a binary vector that
satisfies all the constraints, thus it is a feasible solution to $\MWCDSwIP$.
\end{proof}
\begin{lem}
$\chi$ is a feasible solution to the MWCDS problem. \label{lem:MWCDSw-IP-is-a-MWCDSw}
\end{lem}
\begin{proof}
We first prove that $\chi$ is a DS. Since its equivalent, $\bar{x}^{\ast}$,
is a solution to $\MWCDSwIP$, it must comply with the constraints
$\mbox{\ensuremath{\sum_{v\in N\left(S\right)}x_{v}^{\ast}\geq f\left(S\right)\quad\forall S\in\text{\ensuremath{\mathscr{S}}}}}$.
Let us assume in contradiction that $\chi$ is not a DS. Then, according
to the definition $f\left(\chi\right)=1$. For all nodes $v\in V\setminus\chi$
holds $x_{v}^{\ast}=0$, so $\mbox{\ensuremath{\sum_{v\in N\left(\chi\right)}x_{v}^{\ast}=0}}$.
This is in contradiction to the assumption that $\bar{x}^{\ast}$
satisfies the constraint $\mbox{\ensuremath{\sum_{v\in N\left(\chi\right)}x_{v}^{\ast}\geq f\left(\chi\right)}}$.
Therefore, $\chi$ must be a DS.\\
Now we prove that $\chi$ is a CDS. Since its equivalent, $\bar{x}^{\ast}$,
is an optimal solution to $\MWCDSwIP$, it must comply with all the
constraints $\mbox{\ensuremath{\sum_{v\in N\left(S\right)}x_{v}^{\ast}\geq f\left(S\right)\quad\forall S\in\text{\ensuremath{\mathscr{S}}}}}$
and has the minimum total weight. Let us assume in contradiction that
$\chi$ is not connected. Hence, there are at least two maximal connected
components in the induced subgraph of $\chi$ in $G$. Let us consider
one of them as $C_{\chi}^{1}=\left(V_{C},E_{C}\right)$. In the case
where $f\left(V_{C}\right)=0$, $V_{C}$ is a CDS, and therefore from
Lemma \ref{lem:MWCDSw-is-a-MWCDSw-IP} its equivalent vector satisfies
all the constraints of $\MWCDSwIP$. But, it has a smaller total weight
than $\bar{x}^{\ast}$, and this is in contradiction to the minimality
of $\bar{x}^{\ast}$. In the case where $f\left(V_{C}\right)=1$,
since $C_{\chi}^{1}$ is maximal, then it must be $\mbox{\ensuremath{N\left(V_{C}\right)\cap\chi=\emptyset}}$.
Therefore holds $\mbox{\ensuremath{\sum_{v\in N\left(V_{C}\right)}x_{v}^{\ast}=0}}$.
This is in contradiction to the assumption that $\bar{x}^{\ast}$
satisfies the constraint $\mbox{\ensuremath{\sum_{v\in N\left(V_{C}\right)}x_{v}^{\ast}\geq f\left(V_{C}\right)}}$.
Hence $\chi$ must be connected. To conclude, $\chi$ is dominating
and connected then it must be a CDS.
\end{proof}
\begin{thm}
MWCDS and $\MWCDSwIP$ are equivalent problems.\label{thm:MWCDSw-IP-=00003D-MWCDSw}
\end{thm}
\begin{proof}
From Lemma \ref{lem:MWCDSw-is-a-MWCDSw-IP} $W_{\bar{\sigma}}\geq W_{\bar{x}^{\ast}}$
and from Lemma \ref{lem:MWCDSw-IP-is-a-MWCDSw} $\mbox{\ensuremath{W_{\bar{\sigma}}=OPT_{cds^{\ast}}\leq Weight\left(\chi\right)=W_{\bar{x}^{\ast}}}}$,
therefore must be hold that $\mbox{\ensuremath{OPT_{cds^{\ast}}=W_{\bar{\sigma}}=W_{\bar{x}^{\ast}}}}$.
Since also $\bar{\sigma}$ is a feasible solution for $MWCDS-IP$,
then the statement is true.
\end{proof}
To conclude, let $cds$ be some CDS in $G$. From Theorem \ref{thm:MWCDSw-IP-=00003D-MWCDSw}
and Equation \ref{eq:T-R-P-11} we can state that:
\begin{equation}
Weight\left(cds\right)\geq OPT_{cds^{\ast}}=W_{\bar{x}^{\ast}}\geq F_{\bar{y}}\;.\label{eq:T-R-P-12}
\end{equation}
So, if we find a CDS and a feasible dual solution, we can determine
the approximation ratio between $cds$ and $cds^{\ast}$.
\begin{defn}[Approximation ratio $\alpha$]
A feasible approximation ratio is based on a feasible CDS ($cds$)
and a feasible dual solution ($\bar{y}$), and is defined as follows:
\[
\alpha\triangleq\frac{Weight\left(cds\right)}{F_{\bar{y}}}\geq\frac{Weight\left(cds\right)}{OPT_{cds^{\ast}}}\;.
\]
\label{def:=0003B1}
\end{defn}

\section{The Algorithm\label{sec:The-Algorithm}}

\subsection{Approximation algorithm}

Here, we introduce our algorithm for constructing a data gathering
tree which is based on the $\PD$ method, on Theorem \ref{thm:Gathering-tree-approximation-by-MWCDS}
and on the conclusions of Subsection \ref{subsec:Primal-dual-for-MWCDS}.
The algorithm consists of three main stages (Algorithm \ref{alg:DS},
Algorithm\ref{alg:CDS} and Algorithm \ref{alg:Tree}). Algorithm
\ref{alg:DS} finds an independent DS (IDS), while constructing a
feasible dual solution. Algorithm \ref{alg:CDS} creates a CDS based
on the IDS from the previous stage, and update the dual solution.
Finally, Algorithm \ref{alg:Tree} builds a gathering tree and finds
an ideal MULE location. The BackBone of the gathering tree consists
of the CDS from the previous stage.

\paragraph{Algorithm \ref{alg:DS}}

Receives a connected UDG graph $G=\left(V,E\right)$ with at least
two nodes, and a positive node weight function $\mbox{\ensuremath{w\left(\cdot\right)\colon V\rightarrow R^{+}}}$.
The algorithm returns $ds$ which is an IDS of nodes, and $LB1$ which
is a lower bound value for the weight of the optimal solution to the
MWCDS problem. The algorithm works in steps to find $ds$, while constructing
a feasible dual solution which will serve as the lower bound.

\begin{algorithm}
\begin{algorithmic}[1]

\REQUIRE A connected UDG $G=(V,E)$ with $\left|V\right|\geq2$,
a node weight function $w\left(\cdot\right)\colon V\rightarrow R^{+}$

\ENSURE An independent dominating set $ds$ and a lower bound $LB1$

\STATE $ds\leftarrow\emptyset$~,\quad{}$A\leftarrow\emptyset$~,\quad{}$\text{\ensuremath{\mathscr{S}}}\leftarrow\left\{ \left\{ v\right\} \vcentcolon v\in V\right\} $~,\quad{}$c_{v}\triangleq\frac{1}{100}\cdot w\left(v\right)$\label{algline_DS:l1}

\STATE $y\left(\left\{ v\right\} \right)\leftarrow0\:\forall v\in V$~,\quad{}$g\left(\left\{ v\right\} \right)\leftarrow1\:\forall v\in V$\label{algline_DS:l2}

\REPEAT[The algorithm steps] \label{algline_DS:l3}

\STATE Find node $v\in V$ that minimizes $\epsilon\leftarrow\epsilon\left(v\right)$\label{algline_DS:l4}

\STATE $A\leftarrow A\cup\left\{ v\right\} $\label{algline_DS:l5}

\STATE \textbf{if} $v\notin N\left(ds\right)\cup ds$ \textbf{then}
$ds\leftarrow ds\cup\left\{ v\right\} $\label{algline_DS:l6}

\FOR{each $S\in\text{\ensuremath{\mathscr{S}}}$}\label{algline_DS:l7}

\STATE $y\left(S\right)\leftarrow y\left(S\right)+g\left(S\right)\epsilon$\label{algline_DS:l8}

\STATE \textbf{if} $v\in N\left(S\right)$ \textbf{then} $g\left(S\right)\leftarrow0$\label{algline_DS:l9}

\ENDFOR

\FOR{each $u\in V\setminus A$}\label{algline_DS:l11}

\IF{$\sum_{S\in\text{\ensuremath{\mathscr{S}}}\vcentcolon u\in N\left(S\right)}g\left(S\right)=0$}
\label{algline_DS:l12}

\STATE $S'\leftarrow V\setminus\left\{ u\right\} $\label{algline_DS:l13}

\STATE $\text{\ensuremath{\mathscr{S}}}\leftarrow\text{\ensuremath{\mathscr{S}}}\cup\left\{ S'\right\} $~,\quad{}$y\left(S'\right)\leftarrow0$~,\quad{}$g\left(S'\right)\leftarrow1$\label{algline_DS:l14}

\ENDIF 

\ENDFOR\label{algline_DS:l16}

\UNTIL{$\left(N\left(ds\right)\cup ds=V\right)$}\label{algline_DS:l17}

\RETURN $\left\{ ds\:,\quad LB1\leftarrow\sum_{v\in V}y\left(\left\{ v\right\} \right)\right\} $\label{algline_DS:l18}

\end{algorithmic}

\caption{Constructs an independent dominating set\label{alg:DS}}

\end{algorithm}

Initially, $ds$ is an empty set, the node capacity $c_{v}$ is $\frac{1}{100}\cdot w\left(v\right)$,
and $\text{\ensuremath{\mathscr{S}}}$ consists only of the \textsl{active}
subsets $\mbox{\ensuremath{\left\{ \left\{ v\right\} \vcentcolon v\in V\right\} }}$.
An \textsl{active} subset $S$ is a subset of which the algorithm
is willing to determine its $y_{S}$ value for the dual solution.
For all the other subsets in $\mbox{\ensuremath{\left\{ S\colon S\subseteq V,S\neq\emptyset\right\} }}$,
consider the value $0$. $y\left(S\right)$ is the $y_{S}$ value
determined by the algorithm to subset $S$ at the current step, and
is initialized to $0$. Node $v$ is said to be ``\textit{packed}''
if it satisfies the following equality: $\mbox{\ensuremath{\sum_{S\in\text{\ensuremath{\mathscr{S}}}\colon v\in N\left(S\right)}y\left(S\right)=c_{v}}}$.
Subset $S$ is said to be ``\textit{restricted}'' if it has a \textit{packed}
neighbor node, and therefore its $y\left(S\right)$ value can not
be increased. $g\left(S\right)$ is an indicator with value $0$ for
\textit{restricted} subsets. Formally, it is defined as follows: 
\[
g\left(S\right)=\begin{cases}
0 & if\;\exists v\in V\colon v\in N\left(S\right),\sum_{S'\in\text{\ensuremath{\mathscr{S}}}\vcentcolon v\in N\left(S'\right)}y\left(S'\right)=c_{v}\\
1 & otherwise
\end{cases}\;.
\]
$g\left(S\right)$ is initialized to $1$. The algorithm uses the
potential function $\epsilon\left(\cdot\right)$ to select nodes,
and this function is defined as follows: 
\begin{equation}
\epsilon\left(v\right)=\begin{cases}
\infty & for\;\sum_{S\in\text{\ensuremath{\mathscr{S}}}\colon v\in N\left(S\right)}g\left(S\right)=0\\
\frac{c_{v}-\sum_{S\in\text{\ensuremath{\mathscr{S}}}\colon v\in N\left(S\right)}y\left(S\right)}{\sum_{S\in\text{\ensuremath{\mathscr{S}}}\colon v\in N\left(S\right)}g\left(S\right)} & otherwise
\end{cases}\;.\label{eq:epsilon_v}
\end{equation}

At each step, the algorithm selects the node with the minimum potential
$\epsilon\left(\cdot\right)$ (line \ref{algline_DS:l4}). Let us
say node $v$. If the node is independent in $ds$ then it is added
to $ds$ (line \ref{algline_DS:l6}). The $y\left(S\right)$ value
is uniformly increased by $\epsilon\left(v\right)$ for all \textsl{active}
subsets that are not \textit{restricted} (line \ref{algline_DS:l8}).
Later, we claim that the selected node is \textit{packed} after the
uniform increase of $y\left(S\right)$. Since node $v$ is \textit{packed},
the algorithm updates all subsets that are adjacent to $v$ to be
\textit{restricted} (line \ref{algline_DS:l9}). For each node that
is not yet selected (not \textit{packed}) and all subsets adjacent
to it are \textit{restricted}, the algorithm creates a new subset
in $\text{\ensuremath{\mathscr{S}}}$, which is adjacent only to it
(lines \ref{algline_DS:l11} - \ref{algline_DS:l16}).

The algorithm ends once $ds$ becomes a DS (line \ref{algline_DS:l17}).
The algorithm returns $ds$ and the sum $\mbox{\ensuremath{\sum_{v\in V}y\left(\left\{ v\right\} \right)}}$
as $LB1$ (line \ref{algline_DS:l18}).
\begin{prop}
Each node selected by Algorithm \ref{alg:DS} is packed.\label{prop:alg1-packed}
\end{prop}
\begin{proof}
The algorithm adds $\epsilon\left(v\right)$ to $y\left(S\right)$
of each \textsl{un-restricted} subset $S$ that $v\in N\left(S\right)$.
So we can write the following equation for the selected node $v$:
\[
\begin{gathered}{\displaystyle \sum_{S\in\text{\ensuremath{\mathscr{S}}}\colon v\in N\left(S\right)}}\left(y\left(S\right)+g\left(S\right)\epsilon\left(v\right)\right)=\\
{\displaystyle \sum_{S\in\text{\ensuremath{\mathscr{S}}}\colon v\in N\left(S\right)}}y\left(S\right)+\frac{c_{v}-\sum_{S\in\text{\ensuremath{\mathscr{S}}}\colon v\in N\left(S\right)}y\left(S\right)}{\sum_{S\in\text{\ensuremath{\mathscr{S}}}\colon v\in N\left(S\right)}g\left(S\right)}{\displaystyle \sum_{S\in\text{\ensuremath{\mathscr{S}}}\colon v\in N\left(S\right)}}g\left(S\right)=c_{v}\;.
\end{gathered}
\]
\end{proof}

\paragraph{Algorithm \ref{alg:CDS}}

Receives a connected UDG graph $G=\left(V,E\right)$ with at least
two nodes, a positive node weight function $\mbox{\ensuremath{w\left(\cdot\right)\colon V\rightarrow R^{+}}}$,
and $ds$ which is an IDS in $G$. The algorithm returns $cds$ which
is a CDS of nodes, and $LB2$ which is a lower bound value on the
weight of the optimal solution to the MWCDS problem. The algorithm
works in steps to find $cds$, while constructing a feasible dual
solution which will serve as the lower bound.

\begin{algorithm}
\begin{algorithmic}[1]

\REQUIRE A connected UDG $G=(V,E)$ with $\left|V\right|\geq2$,
a node weight function $w\left(\cdot\right)\colon V\rightarrow R^{+}$,
and an independent dominating set $ds$.

\ENSURE A connected dominating set $cds$ and a lower bound $LB2$.

\STATE $cds\leftarrow ds$~,\quad{}$\text{\ensuremath{\mathscr{S}}}\leftarrow\left\{ \left\{ v\right\} \vcentcolon v\in ds\right\} $~,\quad{}$c_{v}\triangleq\frac{99}{100}\cdot w\left(v\right)$\label{algline_CDS:l1}

\STATE $y\left(\left\{ v\right\} \right)\leftarrow\frac{99}{500}\cdot\left(w\left(v\right)-2\right)\:\forall v\in ds$~,\quad{}$g\left(\left\{ v\right\} \right)\leftarrow1\:\forall v\in ds$
\label{algline_CDS:l2}

\REPEAT[The algorithm steps] \label{algline_CDS:l3}

\STATE Find node $v\in V$ that minimizes $\epsilon\leftarrow\epsilon\left(v\right)$
\label{algline_CDS:l4}

\STATE $\text{\ensuremath{\mathscr{S}}}^{1}\leftarrow\left\{ S\in\text{\ensuremath{\mathscr{S}}}\colon g\left(S\right)=1,\,S\cap N\left(\left\{ v\right\} \right)\neq\emptyset\right\} $\label{algline_CDS:l5}

\STATE $\text{\ensuremath{\mathscr{S}}}^{2}\leftarrow\left\{ S\in\text{\ensuremath{\mathscr{S}}}\colon g\left(S\right)=1,\,S\cap N\left(\left\{ v\right\} \right)\cap cds\neq\emptyset\right\} $
\label{algline_CDS:l6}

\IF{$\sum_{S\in\text{\ensuremath{\mathscr{S}}}\vcentcolon v\in N\left(S\right)}g\left(S\right)>1$}
\label{algline_CDS:l7}

\STATE $cds\leftarrow cds\cup\left\{ v\right\} $ \label{algline_CDS:l8}

\IF{$\text{\ensuremath{\mathscr{S}}}^{1}\setminus\text{\ensuremath{\mathscr{S}}}^{2}\neq\emptyset$}
\label{algline_CDS:l9}

\STATE $S\leftarrow\left\{ u\vcentcolon u\mathrm{\;is\;one\;neighbor\;of\;}v\mathrm{\;from\;each\;}S'\in\text{\ensuremath{\mathscr{S}}}^{1}\setminus\text{\ensuremath{\mathscr{S}}}^{2}\right\} $
\label{algline_CDS:l10}

\STATE $cds\leftarrow cds\cup S$ \label{algline_CDS:l11}

\ENDIF  \label{algline_CDS:l12}

\ENDIF  \label{algline_CDS:l13}

\FOR{each $S\in\text{\ensuremath{\mathscr{S}}}$}\label{algline_CDS:l14}

\STATE $y\left(S\right)\leftarrow y\left(S\right)+g\left(S\right)\epsilon$
\label{algline_CDS:l15}

\STATE \textbf{if} $v\in N\left(S\right)$ \textbf{then} $g\left(S\right)\leftarrow0$
\label{algline_CDS:l16}

\ENDFOR \label{algline_CDS:l17}

\STATE $S\leftarrow\left({\displaystyle \bigcup_{S'\in\text{\ensuremath{\mathscr{S}}}^{1}}}S'\right)\cup\left\{ v\right\} $
\label{algline_CDS:l18}

\STATE $\text{\ensuremath{\mathscr{S}}}\leftarrow\text{\ensuremath{\mathscr{S}}}\cup\left\{ S\right\} $~,\quad{}$y\left(S\right)\leftarrow0$~,\quad{}$g\left(S\right)\leftarrow1$
\label{algline_CDS:l19}

\UNTIL{$\left(\sum_{S\in\text{\ensuremath{\mathscr{S}}}}g\left(S\right)=1\right)$}
\label{algline_CDS:l20}

\RETURN $\left\{ cds\:,\quad LB2\leftarrow\sum_{v\in ds}y\left(\left\{ v\right\} \right)\right\} $\label{algline_CDS:l21}

\end{algorithmic}

\caption{Constructs a connected dominating set\label{alg:CDS}}

\end{algorithm}

Initially, $cds$ contains the $ds$ nodes, the node capacity $c_{v}$
is $\frac{99}{100}\cdot w\left(v\right)$, and $\text{\ensuremath{\mathscr{S}}}$
consists only of the \textsl{active} subsets $\mbox{\ensuremath{\left\{ \left\{ v\right\} \vcentcolon v\in ds\right\} }}$.
$y\left(S\right)$ is initialized to $\mbox{\ensuremath{\frac{99}{500}\cdot\left(w\left(v\right)-2\right)}}$
for each subset $\mbox{\ensuremath{\left\{ v\right\} \vcentcolon v\in ds}}$,
and $g\left(S\right)$ is initialized to $1$.

At each step, the algorithm selects the node with the minimum potential
$\epsilon\left(\cdot\right)$ (line \ref{algline_CDS:l4}). Let us
consider node $v$. If the node is adjacent to at least two \textit{un-restricted}
subsets, then it is added to $cds$ (line \ref{algline_CDS:l8}).
For each subset adjacent to $v$ that does not contain a node from
$N\left(\left\{ v\right\} \right)\cap cds$, the algorithm also adds
to $cds$ a neighbor from that subset (lines \ref{algline_CDS:l10}
- \ref{algline_CDS:l11}). See Figure \ref{fig:Adding-nodes-to-cds}
for illustration of adding nodes. Node $v$ is selected, green nodes
from $ds$, blue from $cds\setminus ds$, and reds are the other nodes
that have been \textit{packed}. In this situation, $v$ and the red
node are added to $cds$. Next, the $y\left(S\right)$ value is uniformly
increased by $\epsilon\left(v\right)$ for all \textsl{active} subsets
that are not \textit{restricted} (line \ref{algline_CDS:l15}). We
later claim that the selected node is \textit{packed}, and therefore
the algorithm updates all subsets that are adjacent to $v$, to be
\textit{restricted} (line \ref{algline_CDS:l16}). The algorithm adds
a new subset to $\text{\ensuremath{\mathscr{S}}}$ containing $v$
and all its adjacent subsets that are not \textit{restricted}.

\begin{figure}
\begin{centering}
\includegraphics[scale=0.8]{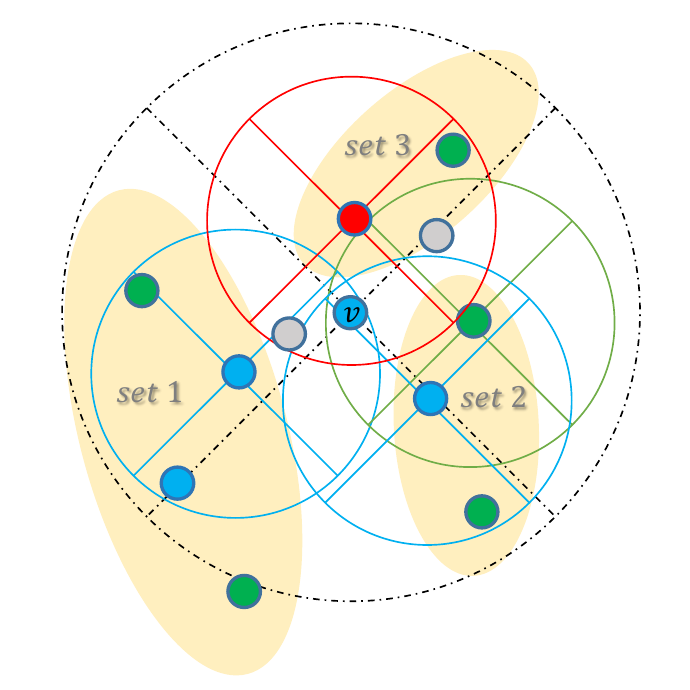}
\par\end{centering}
\caption{\label{fig:Adding-nodes-to-cds}Adding nodes to $cds$ in Algorithm.
\ref{alg:CDS}}

\end{figure}

The algorithm ends once $\text{\ensuremath{\mathscr{S}}}$ contains
only one subset that is not \textit{restricted} (line \ref{algline_CDS:l20}).
The algorithm returns $cds$ and the sum $\mbox{\ensuremath{\sum_{v\in ds}y\left(\left\{ v\right\} \right)}}$
as $LB2$ (line \ref{algline_CDS:l21}).
\begin{prop}
Each node selected by Algorithm \ref{alg:CDS} is packed.\label{prop:alg2-packed}
\end{prop}
\begin{proof}
Stems from the fact that the uniform increase of the $y\left(S\right)$
value is identical to Algorithm \ref{alg:DS}, and from Proposition
\ref{prop:alg1-packed}.
\end{proof}

\paragraph{Algorithm \ref{alg:Tree}}

Receives a connected UDG graph $G=\left(V,E\right)$ with at least
two nodes. The algorithm returns a data gathering tree $T=\left(V,E_{T}\right)$
in graph $G$, the selected MULE location $m$, and the approximation
ratio $\alpha$ (from Theorem \ref{thm:Gathering-tree-approximation-by-MWCDS}).

\begin{algorithm}
\begin{algorithmic}[1]

\REQUIRE A connected UDG $G=(V,E)$ with $\left|V\right|\geq2$ 

\ENSURE A gathering tree $T=\left(V,E_{T}\right)$, the selected
location for the MULE $m$, and the MWCDS approximation ratio $\alpha$

\STATE $E_{T}\leftarrow\emptyset$~,\quad{}$cds\leftarrow\emptyset$

\FOR{each $m'\in V$} \label{algline_Tree:l2}

\STATE $w\left(v\right)\leftarrow\left(2\cdot\dist{m'}v+\C\right)\:\forall v\in V$
\label{algline_Tree:l3}

\STATE $\left\{ ds'\,,\;LB1\right\} \leftarrow$ Algorithm \ref{alg:DS}
$\left(\,G\,,\;w\left(\cdot\right)\,\right)$ \label{algline_Tree:l4}

\STATE $\left\{ cds'\,,\;LB2\right\} \leftarrow$ Algorithm \ref{alg:CDS}
$\left(\,G\,,\;w\left(\cdot\right)\,,\;ds'\,\right)$\label{algline_Tree:l5}

\IF{$\sum_{v\in cds'}w\left(v\right)<\sum_{v\in cds}w\left(v\right)$}

\STATE $m\leftarrow m'$~,\quad{}$LB\leftarrow LB1+LB2$~,\quad{}$cds\leftarrow cds'$
\label{algline_Tree:l7}

\STATE $r\leftarrow$ The first node selected in $ds'$ (at the first
step of Algorithm \ref{alg:DS}) \label{algline_Tree:l8}

\ENDIF 

\ENDFOR \label{algline_Tree:l10}

\STATE $SubG\leftarrow$ The induced subgraph of $cds$ in $G$ \label{algline_Tree:l11}

\STATE $\bar{p}\leftarrow BFS\left(SubG\,,\;r\right)$ \COMMENT{Compute
the predecessors of the nodes $v\in cds$ using BFS algorithm, starting
from the root $r$} \label{algline_Tree:l12}

\FOR{each $v\in cds\setminus\left\{ r\right\} $} \label{algline_Tree:l13}

\STATE $E_{T}\leftarrow E_{T}\cup\left\{ DirectedEdge\left(v\rightarrow p_{v}\right)\right\} $
\COMMENT{$p_{v}$ is the predecessor of node $v$ (a component of
$\bar{p}$)} \label{algline_Tree:l14}

\ENDFOR \label{algline_Tree:l15}

\FOR{each $u\in V\setminus\left(cds\cup\left\{ r\right\} \right)$}
\label{algline_Tree:l16}

\STATE $v\leftarrow\underset{v\in cds}{\arg\min}\left\{ \dist uv\right\} $
\label{algline_Tree:l17}

\STATE $E_{T}\leftarrow E_{T}\cup\left\{ DirectedEdge\left(u\rightarrow v\right)\right\} $
\label{algline_Tree:l18}

\ENDFOR

\RETURN $\left\{ T=\left(V,E_{T}\right)\:,\quad m\:,\quad\alpha\leftarrow\nicefrac{\sum_{v\in cds}w\left(v\right)}{LB}\right\} $\label{algline_Tree:l20}

\end{algorithmic}

\caption{Constructs a data gathering tree\label{alg:Tree}}

\end{algorithm}

The algorithm searches iteratively for the minimal weight CDS produced
by Algorithm \ref{alg:CDS}, over all possible MULE locations. The
node weight function is determined to be $\mbox{\ensuremath{w\left(v\right)=\left(2\cdot\dist mv+\C\right)}}$,
according to Theorem \ref{thm:Gathering-tree-approximation-by-MWCDS}.
The algorithm uses Algorithm \ref{alg:DS} to produce an IDS for Algorithm
\ref{alg:CDS}. The CDS with minimal weight is assigned to $cds$,
the corresponding MULE location is assigned to $m$, $LB$ is the
corresponding sum $LB1+LB2$, and $r$ is the first node selected
for the corresponding $ds'$.

After selecting the minimal weight CDS, the algorithm uses the \textsl{BFS
algorithm} to construct a directed spanning tree towards the root
node $r$, in the induced subgraph of $cds$ in $G$ (lines \ref{algline_Tree:l11}
- \ref{algline_Tree:l15}). In this way, the $cds$ nodes compose
the BackBone of the gathering tree. Then, each node in $\mbox{\ensuremath{V\setminus cds}}$
is connected by an outgoing edge to the nearest node in $cds$.

Finally, the algorithm returns the result tree as $T$, the MULE location
$m$, and the approximation ratio $\alpha$ of the result.

\subsection{Correctness}

In this section, we will show the correctness of our algorithms. We
will emphasize in advance that the algorithms are required to receive
a connected UDG graph with at least two nodes in order to return correct
structures of IDS, CDS, and gathering tree. However, for the approximation
ratio and for the lower bounds, additional constraints are required.
For correct lower bounds ($LB1$ and $LB2$) and the $\alpha$ parameter,
the graph should have a diameter of at least three. For correct approximation
of the $\MULE$ problem (i.e. $\left(1+\varepsilon\right)$), it is
required to be hold $\underset{v\in V_{BB}^{\ast}}{Average}\left[\dist{m^{\ast}}v\right]>1.3$,
according to Theorem \ref{thm:Gathering-tree-approximation-by-MWCDS}.
As for graphs that do not meet only the first condition (on the diameter),
these graphs are spread over a bounded area, and we will address this
situation later in Subsection \ref{subsec:Performance}. We will show
in a different approach that the result gathering tree still ensures
the approximation required. As for graphs that do not meet the second
condition, this issue still remains open.
\begin{defn}[A feasible dual solution and lower bound]
Consider the MWCDS problem and the dual problem $\MWCDSwD$. Vector
$\bar{y}$ will be called a \textsl{feasible dual solution}, if $\bar{y}$
is proven to be a solution to the $\MWCDSwD$ problem and satisfies
all its constraints. A \textit{feasible lower bound} is a value that
has been proven to have a \textit{feasible }\textsl{dual solution}
which is well defined, whose sum $\sum_{S\subseteq V\vcentcolon S\neq\emptyset}f\left(S\right)\cdot y_{S}$
is equal to this value. $y_{S}$ are the components of $\bar{y}$,
$V$ is the set of the graph nodes and $f\left(S\right)$ is defined
in Equation \ref{eq:f}.\label{def:feasible_LB}
\end{defn}

\paragraph{Correctness of Algorithm \ref{alg:DS}.}

Let $\text{\ensuremath{\mathscr{S}}}'$ be the collection of all \textsl{active}
subsets that Algorithm \ref{alg:DS} maintains in the current step
(i.e. $\text{\ensuremath{\mathscr{S}}}'\subseteq\text{\ensuremath{\mathscr{S}}}$),
and let $y'\left(S\right)$ be the value the algorithm assigns to
subset $S$ in the current step. We define the vector $\bar{y}'$
of size $\left|\text{\ensuremath{\mathscr{S}}}\right|$ (where $y'_{S}$
are its components) as follows: 
\begin{equation}
y'_{S}=\begin{cases}
y'\left(S\right) & for\;\left|S\right|=1\\
0 & for\;\left|S\right|>1
\end{cases}\quad\forall S\in\text{\ensuremath{\mathscr{S}}}\,.\label{eq:y'S}
\end{equation}
$\epsilon$ is the potential value of the selected node in the current
step. $A$ is the set of nodes that the algorithm has selected up
to the current step. From Lemma \ref{lem:IDS}, Algorithm \ref{alg:DS}
returns an IDS. From Lemma \ref{lem:LB1}, if the graph diameter is
at least three then it returns a feasible lower bound.
\begin{prop}
Once node $v$ is selected by Algorithm \ref{alg:DS} it becomes packed
and $\epsilon\left(v\right)=\infty$, and it remains so until the
end.\label{prop:alg1-remains-packed}
\end{prop}
\begin{proof}
From Proposition \ref{prop:alg1-packed}, after node $v$ is chosen
by the algorithm it is \textsl{packed}. According to lines \ref{algline_DS:l7}
and \ref{algline_DS:l9}, after node $v$ is chosen, there is $\mbox{\ensuremath{\sum_{S\in\text{\ensuremath{\mathscr{S}}}'\colon v\in N\left(S\right)}g\left(S\right)=0}}$,
and therefore by Equation \ref{eq:epsilon_v} $\epsilon\left(v\right)=\infty$.
The only place that can change the packaging condition or the $\infty$
potential value is in line \ref{algline_DS:l14}, since only here
the algorithm performs $g\left(S'\right)\leftarrow1$. But this assignment
is performed only for new subsets. For both conditions, subset $S'$
will influence node $v$ only if $v\in N\left(S'\right)$. New subsets
are of the form $V\setminus\left\{ u\right\} $, and they have only
one neighbor which is $u$. Since node $v$ is already selected and
added to $A$, it is not contained in the set $V\setminus A$. Hence,
according to line \ref{algline_DS:l11}, the subset $V\setminus\left\{ v\right\} $
will not be added later in the algorithm. Therefore, new subsets will
not influence node $v$, which will remain \textsl{packed} and $\epsilon\left(v\right)=\infty$.
\end{proof}
\begin{prop}
In Algorithm \ref{alg:DS}, if $g\left(S\right)=0$ for all $S\in\text{\ensuremath{\mathscr{S}}}'$,
then there must be $A=V$.\label{prop:A=00003DV}
\end{prop}
\begin{proof}
Let us assume in contradiction that $\mbox{\ensuremath{g\left(S\right)=0\;\forall S\in\text{\ensuremath{\mathscr{S}}}'}}$,
but $A\neq V$. If so, there must be a node $u\in V\setminus A$,
and from lines \ref{algline_DS:l12} - \ref{algline_DS:l14}, there
must be a subset $V\setminus\left\{ u\right\} \in\text{\ensuremath{\mathscr{S}}}'$
which is adjacent to $u$. Since $u$ is the only neighbor of $V\setminus\left\{ u\right\} $,
and $u$ is not yet selected by the algorithm, then the \textit{if}
statement in line \ref{algline_DS:l9} never met, nor does the assignment
$g\left(V\setminus\left\{ u\right\} \right)\leftarrow0$. This is
in contradiction to the assumption $\mbox{\ensuremath{g\left(S\right)=0\;\forall S\in\text{\ensuremath{\mathscr{S}}}'}}$,
and therefore must be $A=V$.
\end{proof}
\begin{prop}
The set $ds$ constructed by Algorithm \ref{alg:DS} is an IS, and
if $A=V$ then $ds$ must be a MIS.\label{prop:MIS}
\end{prop}
\begin{proof}
The set $ds$ is IS, since, in line \ref{algline_DS:l6} a node is
added to $ds$ only if it is independent of the $ds$ nodes. Note
that by line \ref{algline_DS:l5}, if there is a step where $A=V$,
it requires that each node was selected at least once. Let us assume
in contradiction that $A=V$, but the set $ds$ is not a MIS. Then
there is at least one node $v\in V\setminus ds$ that can be added
to $ds$, so it must be $v\notin N\left(ds\right)\cup ds$. But since
$v$ is already chosen by the algorithm, it is in contradiction to
the condition in line \ref{algline_DS:l6}. Hence, $ds$ must be MIS.
\end{proof}
\begin{prop}
Each node is selected at most once during Algorithm \ref{alg:DS}.\label{prop:alg1-node-selected-once}
\end{prop}
\begin{proof}
According to Proposition \ref{prop:alg1-remains-packed}, the potential
value of each node $v\in V$ selected by the algorithm becomes $\epsilon\left(v\right)=\infty$,
and remains so until the end. We want to show that there is no step
in which a node is selected with a potential value of $\epsilon\left(v\right)=\infty$.
To do this, let us assume in contradiction that there is a step where
node $v$ is selected with $\epsilon\left(v\right)=\infty$, but the
stop condition (line \ref{algline_DS:l17}) has not yet met. Then,
according to line \ref{algline_DS:l4}, for each node $u\in V$ holds
$\epsilon\left(u\right)=\infty$. Since $c_{u}<\infty$, then it must
be the first case in Equation \ref{eq:epsilon_v}. The graph is connected
with at least two nodes, so each subset in $\text{\ensuremath{\mathscr{S}}}'$
is adjacent to at least one node. Therefore, must be $g\left(S\right)=0$
for all $S\in\text{\ensuremath{\mathscr{S}}}'$, and according to
Proposition \ref{prop:A=00003DV}, must also hold $A=V$. Consequently,
from Proposition \ref{prop:MIS}, the set $ds$ must be a MIS, hence
it is also a DS. This is in contradiction to the assumption that $\mbox{\ensuremath{N\left(ds\right)\cup ds\neq V}}$.
Therefore, the claim is correct and node can not be selected twice.
\end{proof}
\begin{lem}
Given a connected UDG with at least two nodes and node weights $\mbox{\ensuremath{w\left(\cdot\right)\colon V\rightarrow R^{+}}}$,
Algorithm \ref{alg:DS} returns an IDS.\label{lem:IDS}
\end{lem}
\begin{proof}
First, we argue that the set $ds$ contains at least one node. Since
the graph is connected and has at least two nodes, and $\text{\ensuremath{\mathscr{S}}}'$
is initialized to be $\mbox{\ensuremath{\left\{ \left\{ v\right\} \vcentcolon v\in V\right\} }}$,
then at the first iteration, each node $v\in V$ has at least one
subset $S\in\text{\ensuremath{\mathscr{S}}}'$ for which $v\in N\left(S\right)$.
Without loss of generality, let us consider node $v$ and some subset
$S$, where $v\in N\left(S\right)$. From line \ref{algline_DS:l2},
$g\left(S\right)\leftarrow1$, so for node $v$ holds: $\mbox{\ensuremath{\sum_{S\in\text{\ensuremath{\mathscr{S}}}'\colon v\in N\left(S\right)}g\left(S\right)>0}}$.
Since also $\mbox{\ensuremath{c_{v}=\frac{1}{100}\cdot w\left(v\right)<\infty}}$,
then from Equation \ref{eq:epsilon_v}, at least node $v$ will be
selected at the first iteration, and $\mbox{\ensuremath{\epsilon}\ensuremath{\ensuremath{\leq\epsilon\left(v\right)}<\ensuremath{\infty}}}$.

From line \ref{algline_DS:l17}, we can see that the algorithm halts
once set $ds$ becomes a DS. In Proposition \ref{prop:alg1-node-selected-once},
we saw that each node can be selected at most once. Then the iterations
must be halt, and this is as a result of the fact that the set $ds$
has become a DS.

Combine all that with Proposition \ref{prop:MIS}, we get that $ds$
is both IS and DS. Therefore it is an IDS, and hence the lemma is
correct.
\end{proof}
\begin{lem}
Given a connected UDG with a diameter of at least three and node weights
$\mbox{\ensuremath{w\left(\cdot\right)\colon V\rightarrow R^{+}}}$,
Algorithm \ref{alg:DS} returns a feasible lower bound, having $\bar{y}'$
as its feasible dual solution that also satisfies \linebreak{}
$\mbox{\ensuremath{\sum_{S\in\text{\ensuremath{\mathscr{S}}}\colon v\in N\left(S\right)}y'_{S}\leq\frac{1}{100}\cdot w\left(v\right)\;\forall v\in V}}$.\label{lem:LB1}
\end{lem}
\begin{proof}
We will prove by induction on the algorithm steps that it maintains
the constraints \linebreak{}
$\mbox{\ensuremath{\sum_{S\in\text{\ensuremath{\mathscr{S}}}'\colon v\in N\left(S\right)}y'\left(S\right)\leq c_{v}\;\forall v\in V}}$.
Since at the initialization phase $\mbox{\ensuremath{y'\left(S\right)=0\quad\forall S\in\text{\ensuremath{\mathscr{S}}}'}}$
(line \ref{algline_DS:l2}) and since $c_{v}>0$, then it is easy
to see that the constraints are maintained for base case $step=0$.
The induction hypothesis is that the statement is correct up to $step=k$,
and we will prove that the statement remains correct for $step=k+1$.
Note, the only line that can violate the constraints is line \ref{algline_DS:l8},
where $y'\left(S\right)$ is updated. In step $k+1$, the algorithm
selects node $v$. Therefore, according to Proposition \ref{prop:alg1-packed},
node $v$ is \textit{packed}, i.e. $\mbox{\ensuremath{\sum_{S\in\text{\ensuremath{\mathscr{S}}}'\colon v\in N\left(S\right)}y'\left(S\right)=c_{v}}}$,
and according to Proposition \ref{prop:alg1-remains-packed}, it remains
\textit{packed} until the end of the algorithm. Then, nodes where
the statement can be violated are only those that have not yet been
selected (i.e. from $V\setminus A$). Let us assume in contradiction
that there is a node $u\in V\setminus A$ that after step $k+1$ violates
the statement, and let $y'_{k}\left(S\right)$ be the value of $y'\left(S\right)$
at step $k$. We can write: 
\[
\begin{array}{rlc}
\sum_{S\in\text{\ensuremath{\mathscr{S}}}'\colon u\in N\left(S\right)}y'\left(S\right)> & c_{u}\\
\sum_{S\in\text{\ensuremath{\mathscr{S}}}'\colon u\in N\left(S\right)}\left(y'_{k}\left(S\right)+g\left(S\right)\epsilon\left(v\right)\right)> & c_{u}\\
\epsilon\left(v\right)> & \frac{c_{u}-\sum_{S\in\text{\ensuremath{\mathscr{S}}}'\colon u\in N\left(S\right)}y'_{k}\left(S\right)}{\sum_{S\in\text{\ensuremath{\mathscr{S}}}'\colon u\in N\left(S\right)}g\left(S\right)}= & \epsilon\left(u\right)\;.
\end{array}
\]
The result is in contradiction to the minimality of $\epsilon\left(v\right)$.
So, the statement is correct for step $k+1$. Hence, the algorithm
satisfies the constraints.

$\text{\ensuremath{\mathscr{S}}}'$ contains subsets $\mbox{\ensuremath{\left\{ v\right\} \vcentcolon v\in V}}$
(line \ref{algline_DS:l1}) and $\mbox{\ensuremath{y'_{S}=0\;\forall S:\left|S\right|>1}}$
(Equation \ref{eq:y'S}), then we get \linebreak{}
$\mbox{\ensuremath{\sum_{S\in\text{\ensuremath{\mathscr{S}}}\colon v\in N\left(S\right)}y'_{S}=\sum_{S\in\text{\ensuremath{\mathscr{S}}}'\colon v\in N\left(S\right)}y'_{S}}}$.
Since $\mbox{\ensuremath{\sum_{S\in\text{\ensuremath{\mathscr{S}}}'\colon v\in N\left(S\right)}y'\left(S\right)\leq c_{v}}}$,
then $\mbox{\ensuremath{\epsilon\left(v\right)\geq0}}$ (Equation
\ref{eq:epsilon_v}). Combine it with line \ref{algline_DS:l8}, we
get $\mbox{\ensuremath{y'\left(S\right)\geq0}}$. So, from Equation
\ref{eq:y'S}, $\mbox{\ensuremath{y'_{S}\leq y'\left(S\right)}}$.
Recall that $\mbox{\ensuremath{c_{v}\triangleq\frac{1}{100}\cdot w\left(v\right)}}$.
Then the following holds: 
\[
\sum_{S\in\text{\ensuremath{\mathscr{S}}}\colon v\in N\left(S\right)}y'_{S}\leq\sum_{S\in\text{\ensuremath{\mathscr{S}}}'\colon v\in N\left(S\right)}y'\left(S\right)\leq c_{v}=\frac{1}{100}\cdot w\left(v\right)\leq w\left(v\right)\;\forall v\in V\,.
\]
Since also $\mbox{\ensuremath{y'_{S}\geq0}}$ (Equation \ref{eq:y'S}),
then vector $\bar{y}'$ satisfies all the constraints of $\MWCDSwD$.
Hence, $\bar{y}'$ is a feasible dual solution.

Since the diameter of the graph is at least three, no subset of a
single node is a DS. That is, $\mbox{\ensuremath{f\left(\left\{ v\right\} \right)=1\;\forall v\in V}}$
(Equation \ref{eq:f}). Combine it with Equation \ref{eq:y'S} and
line \ref{algline_DS:l18}, we get 
\begin{equation}
\sum_{S\subseteq V\vcentcolon S\neq\emptyset}f\left(S\right)\cdot y'_{S}=\sum_{S\subseteq V\vcentcolon S\neq\emptyset}y'_{S}=\sum_{S\subseteq V\colon\left|S\right|=1}y'\left(S\right)=\sum_{v\in V}y'\left(\left\{ v\right\} \right)=LB1\,.\label{eq:LB1}
\end{equation}
Therefore, by Definition \ref{def:feasible_LB}, $LB1$ is a feasible
lower bound and $\bar{y}'$ is its feasible dual solution. Hence,
the lemma is correct.
\end{proof}

\paragraph{Correctness of Algorithm \ref{alg:CDS}.}

Let $\text{\ensuremath{\mathscr{S}}}''$ be the collection of the
\textsl{active} subsets that Algorithm \ref{alg:CDS} maintains in
the current step (i.e. $\text{\ensuremath{\mathscr{S}}}''\subseteq\text{\ensuremath{\mathscr{S}}}$),
and let $y''\left(S\right)$ be the value that the algorithm assigns
to subset $S$ in the current step. We define the vector $\bar{y}''$
of size $\left|\text{\ensuremath{\mathscr{S}}}\right|$ (where $y''_{S}$
are its components) as follows: 
\begin{equation}
y''_{S}=\begin{cases}
y''\left(S\right) & for\;S=\left\{ v\right\} \vcentcolon v\in ds\\
0 & otherwise
\end{cases}\quad\forall S\in\text{\ensuremath{\mathscr{S}}}\,.\label{eq:y''S}
\end{equation}
$\epsilon$ is the potential value of the selected node in the current
step. Let $\text{\ensuremath{\mathscr{S}}}^{1}$ be the collection
of all \textsl{un-restricted} \textsl{active} subsets that contain
some neighbor of the selected node, and $\text{\ensuremath{\mathscr{S}}}^{2}$
be the collection of all subsets in $\text{\ensuremath{\mathscr{S}}}^{1}$
that also contain a node which is both a neighbor of the selected
node and belongs to the $cds$ set. From Lemma \ref{lem:CDS}, Algorithm
\ref{alg:CDS} returns a CDS. From Lemma \ref{lem:LB2}, if the graph
diameter is at least three and node weights are $\mbox{\ensuremath{w\left(v\right)=2\cdot\dist mv+\C}}$,
then it returns a feasible lower bound.
\begin{prop}
At each step of Algorithm \ref{alg:CDS}, if node $v$ is contained
in some active subset, then there must be an un-restricted subset
which contains $v$. \label{prop:g(v)>0}
\end{prop}
\begin{proof}
We will prove by induction on the steps, for each node separately.
Without loss of generality, let $v$ be the node. The base case is
the earliest step in which $v$ is contained in a subset. Let $k_{0}$
be this step and $S_{0}$ be this subset. If $v\in ds$, then by lines
\ref{algline_CDS:l1} and \ref{algline_CDS:l2}, the base case is
trivial. If $v\notin ds$, then by line \ref{algline_CDS:l18}, $v$
must be selected by the algorithm to be contained in $S_{0}$. By
line \ref{algline_CDS:l19}, the condition holds for the base case.
If $v$ is never selected, then it is irrelevant. The induction hypothesis
is that the statement holds for step $k>k_{0}$, and we will prove
that the statement still holds for step $k+1$. According to the induction
hypothesis, at step $k$ there is an \textsl{un-restricted} subset
$S'$ which contains $v$. Let $u$ be the node selected at step $k+1$.
If $u\notin N\left(S'\right)$, then by line \ref{algline_CDS:l16},
the situation remains the same. Let $S$ be the new subset at step
$k+1$. If $u\in N\left(S'\right)$, then $S'\in\text{\ensuremath{\mathscr{S}}}^{1}$
(line \ref{algline_CDS:l5}) and therefore $v\in S$ (line \ref{algline_CDS:l18}).
In line \ref{algline_CDS:l19}, $g\left(S\right)\leftarrow1$, so
the statement holds. Hence the claim is correct.
\end{proof}
\begin{prop}
After node $v$ is selected by Algorithm \ref{alg:CDS} it becomes
packed and $\epsilon\left(v\right)=\infty$. Once a node becomes packed
or the potential value becomes $\infty$, it remains so until the
end. \label{prop:alg2-remains-packed}
\end{prop}
\begin{proof}
From Proposition \ref{prop:alg2-packed}, after node $v$ is selected
by the algorithm it is \textsl{packed}. According to lines \ref{algline_CDS:l14}
and \ref{algline_CDS:l16}, after node $v$ is selected, there is
$\mbox{\ensuremath{\sum_{S\in\text{\ensuremath{\mathscr{S}}}''\colon v\in N\left(S\right)}g\left(S\right)=0}}$,
and therefore by Equation \ref{eq:epsilon_v}, $\epsilon\left(v\right)=\infty$.
The only place that can change the \textsl{packing} condition or the
$\infty$ potential value is in line \ref{algline_CDS:l19}, since
only here the algorithm performs $g\left(S\right)\leftarrow1$. But
this assignment is performed only for new subsets. For both conditions,
subset $S$ will influence node $v$ only if $v\in N\left(S\right)$
and the assignment $g\left(S\right)\leftarrow1$ will be performed.
Let us assume in contradiction that there is a new subset that influences
$v$ at some later step. Let $k$ be the earliest step in which $v$
is influenced. Let $S$ be the new subset created in this step, and
$u$ be the selected node. $S$ consists of all nodes of subsets $\mbox{\ensuremath{S'\in\text{\ensuremath{\mathscr{S}}}''\vcentcolon g\left(S'\right)=1,S'\cap N\left(\left\{ u\right\} \right)\neq\emptyset}}$
and node $u$. None of these subsets can be adjacent to $v$ at step
$k-1$, since $\epsilon\left(v\right)=\infty$. Therefore, the only
possibility to influence $v$ is if $u$ and $v$ are neighbors. From
Proposition \ref{prop:g(v)>0}, there is a subset $\mbox{\ensuremath{S''\in\text{\ensuremath{\mathscr{S}}}''\vcentcolon v\in S'',g\left(S''\right)=1}}$.
By line \ref{algline_CDS:l5}, $S''\in\text{\ensuremath{\mathscr{S}}}^{1}$,
and therefore $v\in S$. This is in contradiction to the assumption
that $S$ influences $v$. Hence the claim is correct.
\end{proof}
\begin{prop}
Algorithm \ref{alg:CDS} does not select a node with potential value
$\infty$. \label{prop:chosen-not-infinity}
\end{prop}
\begin{proof}
Let us assume in contradiction that there is a step in which the algorithm
selects a node with potential value $\infty$. Consider $k$ to be
the earliest step at which such a node is selected, and let $v$ be
this node. By line \ref{algline_CDS:l4}, must hold $\mbox{\ensuremath{\epsilon\left(u\right)=\infty}}$
for all $u\in V$. From Equation \ref{eq:epsilon_v}, and since $c_{v}<\infty$,
it requires that $\mbox{\ensuremath{\sum_{S\in\text{\ensuremath{\mathscr{S}}}''\colon u\in N\left(S\right)}g\left(S\right)=0}}$
for all $u\in V$. Since the graph is connected, it implies that there
is no \textsl{un-restricted} subset $S\in\text{\ensuremath{\mathscr{S}}}''\vcentcolon\left|S\right|<\left|V\right|$.
According to line \ref{algline_CDS:l19}, we know that there must
be an \textsl{un-restricted} subset $\mbox{\ensuremath{S'\in\text{\ensuremath{\mathscr{S}}}''}}$.
Then this subset must be $S'=V$. So it also has to be the only \textsl{un-restricted}
subset. This is in contradiction to the stop condition (line \ref{algline_CDS:l20})
at the end of step $k-1$. Hence the claim is correct.
\end{proof}
\begin{prop}
At each step of Algorithm \ref{alg:CDS}, the union of all un-restricted
active subsets contains the set $cds$. \label{prop:union-contains-CDS}
\end{prop}
\begin{proof}
By line \ref{algline_CDS:l1}, nodes from $ds$ are contained in some
\textsl{active} subset $S$. As for nodes from $cds\setminus ds$,
each node had to be selected by the algorithm to be added to $cds$.
Thus, by line \ref{algline_CDS:l18}, this node was also added to
some \textsl{active} subset. According to Proposition \ref{prop:g(v)>0},
a node that is contained in some \textsl{active} subset must also
be contained in an \textsl{un-restricted} subset. Hence, all the nodes
in $cds$ are included in the union.
\end{proof}
\begin{prop}
During Algorithm \ref{alg:CDS}, no node from subset $ds$ will be
selected.\label{prop:v-not-in-DS}
\end{prop}
\begin{proof}
Let $v$ be some node from $ds$. Since $ds$ is an IS, at the initialization
stage (line \ref{algline_CDS:l1}) there is no subset $S\in\text{\ensuremath{\mathscr{S}}}''$
adjacent to $v$. So, $\mbox{\ensuremath{\sum_{S\in\text{\ensuremath{\mathscr{S}}}''\colon v\in N\left(S\right)}g\left(S\right)=0}}$,
and as a result $\mbox{\ensuremath{\epsilon\left(v\right)=\infty}}$.
From propositions \ref{prop:alg2-remains-packed} and \ref{prop:chosen-not-infinity},
the potential value of $v$ remains $\infty$ until the end, and therefore
the node will not be selected by the algorithm.
\end{proof}
\begin{prop}
For each node $v\in S\vcentcolon S\in\text{\ensuremath{\mathscr{S}}}''$,
if $v\notin cds$ then it must have a neighbor $\mbox{\ensuremath{u\in S\cap cds}}$.
\label{prop:neighbor-in-CDS}
\end{prop}
\begin{proof}
Consider some node $v$ and some subset $S$, where $S\in\text{\ensuremath{\mathscr{S}}}''$
and $v\in S\setminus cds$. Of course $v\notin ds$, consequently
$v$ was not added to $S$ at the initialization stage (step $0$),
but rather in a later step, say step $k>0$. Note that from propositions
\ref{prop:alg2-packed}, \ref{prop:alg2-remains-packed} and \ref{prop:chosen-not-infinity},
after node $v$ is selected, $\epsilon\left(v\right)=\infty$ until
the end of the algorithm, and so it will not be selected again. Without
loss of generality, let $S$ be the first subset $v$ added to. Otherwise,
since $v$ will not be selected again, by line \ref{algline_CDS:l18}
all other subsets containing $v$ also contain all $S$ nodes. Since
$ds$ is a DS and $v\notin ds$, then there is a node $\mbox{\ensuremath{u\in ds\cap N\left(\left\{ v\right\} \right)}}$.
Note that in step $0$, $u$ was contained in some subset, and therefore,
according to Proposition \ref{prop:g(v)>0}, in step $k-1$, there
is a subset $\mbox{\ensuremath{S'\in\text{\ensuremath{\mathscr{S}}}''\vcentcolon u\in S',g\left(S'\right)=1}}$.
In step $k$, $S'$ will be included in the collection $\text{\ensuremath{\mathscr{S}}}^{1}$.
We assumed that $v\notin cds$, so it must be because the condition
in line \ref{algline_CDS:l7} is not met. Thus it requires that $\left|\text{\ensuremath{\mathscr{S}}}^{1}\right|=1$,
and only $S'$ contained in $\text{\ensuremath{\mathscr{S}}}^{1}$.
Recall that the neighbor $\mbox{\ensuremath{u\in S'\cap ds}}$, hence
$\mbox{\ensuremath{u\in S'\cap cds}}$. According to line \ref{algline_CDS:l18},
$S$ contains $v$ and all $S'$ nodes, and therefore the claim is
correct.
\end{proof}
\begin{lem}
Given a connected UDG with at least two nodes, node weights $\mbox{\ensuremath{w\left(\cdot\right)\colon V\rightarrow R^{+}}}$,
and an IDS, Algorithm \ref{alg:CDS} returns a CDS.\label{lem:CDS}
\end{lem}
\begin{proof}
First, we will show that the algorithm halts. From Proposition \ref{prop:alg2-packed},
the selected node is packed, and according to Proposition \ref{prop:alg2-remains-packed},
it remains so until the end. Thus after a node is selected, its potential
value will remain $\infty$. According to Proposition \ref{prop:chosen-not-infinity},
this node will not be selected again. Hence, the algorithm must halts
after at most $\left|V\right|$ steps.

We will show that the resulting $cds$ is indeed CDS. Since $cds$
contains $ds$ then it is dominates all $V$ nodes.

After the algorithm halts, by line \ref{algline_CDS:l20}, there must
be $\mbox{\ensuremath{\sum_{S\in\text{\ensuremath{\mathscr{S}}}''}g\left(S\right)=1}}$.
According to Proposition \ref{prop:union-contains-CDS}, the only
\textsl{un-restricted} subset $S\in\text{\ensuremath{\mathscr{S}}}''$
must contain $cds$. We will prove by induction on the algorithm steps
that it maintains the intersection $S'\cap cds$ connected for each
$S'\in\text{\ensuremath{\mathscr{S}}}''$ (and especially for $S$).
\\
According to line \ref{algline_CDS:l1}, initially, each subset $S'\in\text{\ensuremath{\mathscr{S}}}''$
contains a single node from $ds$. So the base case is trivial. The
induction hypothesis is that the statement holds for step $k$, and
we will prove that the statement still holds for step $k+1$. The
only new subset added during step $k+1$ is in line \ref{algline_CDS:l19}.
Consider $S'$ to be the new subset and $v$ to be the selected node.
$ds$ is a DS, and by Proposition \ref{prop:v-not-in-DS}, we know
that $v$ is not in $ds$. Then, there is a node $u\in ds$ which
is $u\in N\left(\left\{ v\right\} \right)$. By line \ref{algline_CDS:l1},
$u$ is contained in some subset, so according to Proposition \ref{prop:g(v)>0},
there is a subset $\mbox{\ensuremath{S''\vcentcolon u\in S''}}$ that
satisfies $g\left(S''\right)=1$. Therefore, $\text{\ensuremath{\mathscr{S}}}^{1}$
contains at least the subset $S''$.\\
If the condition in line \ref{algline_CDS:l7} is not met, then $v$
is not added to $cds$, and $\text{\ensuremath{\mathscr{S}}}^{1}$
contains only $S''$. From the induction hypothesis the statement
holds, since $\mbox{\ensuremath{S'\cap cds=\left(S''\cup\left\{ v\right\} \right)\cap cds=S''\cap cds}}$.\\
If the condition in line \ref{algline_CDS:l7} is met, then $v$ is
added to $cds$ and $\left|\text{\ensuremath{\mathscr{S}}}^{1}\right|>1$.
We will show that $\mbox{\ensuremath{\left(S''\cup\left\{ v\right\} \right)\cap cds}}$
is CS, for all $S''\in\text{\ensuremath{\mathscr{S}}}^{1}$. Therefore,
by line \ref{algline_CDS:l18}, $\mbox{\ensuremath{S'\cap cds}}$
must also be connected. Note that $\text{\ensuremath{\mathscr{S}}}^{2}\subseteq\text{\ensuremath{\mathscr{S}}}^{1}$.
If $S''\in\text{\ensuremath{\mathscr{S}}}^{2}$, then $v$ has a neighbor
$\mbox{\ensuremath{u\in S''\cap cds}}$. Hence $\mbox{\ensuremath{\left(S''\cup\left\{ v\right\} \right)\cap cds}}$
is also CS. If $S''\in\text{\ensuremath{\mathscr{S}}}^{1}\setminus\text{\ensuremath{\mathscr{S}}}^{2}$,
then by lines \ref{algline_CDS:l9} - \ref{algline_CDS:l12}, the
algorithm adds to $cds$ some neighbor of $v$. Let $u$ be this neighbor,
and note that $u\notin cds$. From Proposition \ref{prop:neighbor-in-CDS},
$u$ has a neighbor in $S''\cap cds$, then $\mbox{\ensuremath{S''\cap\left(cds\cup\left\{ u\right\} \right)}}$
is CS. After line \ref{algline_CDS:l11}, $S''\cap cds$ remains connected,
and so is $\mbox{\ensuremath{\left(S''\cup\left\{ v\right\} \right)\cap cds}}$.\\
The statement holds for step $k+1$. Therefore, $S'\cap cds$ is a
CS for each $S'\in\text{\ensuremath{\mathscr{S}}}''$.

If we consider the last step, $S\cap cds$ is also a CS. Since $\mbox{\ensuremath{S\cap cds=cds}}$,
then $cds$ must also be a CDS. Therefore, the lemma is correct.
\end{proof}
\begin{lem}
Given a connected UDG with a diameter of at least three and node weights
$\mbox{\ensuremath{w\left(v\right)=2\cdot\dist mv+\C}}$, Algorithm
\ref{alg:CDS} returns a feasible lower bound, having $\bar{y}''$
as its feasible dual solution that also satisfies \linebreak{}
$\mbox{\ensuremath{\sum_{S\in\text{\ensuremath{\mathscr{S}}}\colon v\in N\left(S\right)}y''_{S}\leq\frac{99}{100}\cdot w\left(v\right)\;\forall v\in V}}$.\label{lem:LB2}
\end{lem}
\begin{proof}
We will prove by induction on the algorithm steps that it maintains
the constraints \linebreak{}
$\mbox{\ensuremath{\sum_{S\in\text{\ensuremath{\mathscr{S}}}''\colon v\in N\left(S\right)}y''\left(S\right)\leq c_{v}\;\forall v\in V}}$.
In the initialization phase $\text{\ensuremath{\mathscr{S}}}''\triangleq\left\{ \left\{ v\right\} \vcentcolon v\in ds\right\} $,
so \linebreak{}
$\mbox{\ensuremath{\sum_{S\in\text{\ensuremath{\mathscr{S}}}''\colon v\in N\left(S\right)}y''\left(S\right)=\sum_{u\in ds\colon v\in N\left(\left\{ u\right\} \right)}y''\left(\left\{ u\right\} \right)}}$.
Recall that $y''\left(\left\{ u\right\} \right)$ are initialized
to $\mbox{\ensuremath{\frac{99}{500}\cdot\left(w\left(u\right)-2\right)}}$.
Since node weights are $\mbox{\ensuremath{w\left(u\right)\triangleq2\cdot\dist mu+\C}}$
and the distance between a pair of neighbors in UDG is at most $1$,
then $\mbox{\ensuremath{w\left(u\right)\leq w\left(v\right)+2\vcentcolon v\in N\left(\left\{ u\right\} \right)}}$.
$ds$ is an IS in UDG. Therefore, each node in the graph can have
at most $5$ neighbors in $ds$ (the kissing number problem). Formally,
$\mbox{\ensuremath{\left|ds\cap N\left(\left\{ v\right\} \right)\right|\leq5\;\forall v\in V}}$.
When we put it all together:

\[
\begin{array}{rl}
{\displaystyle \sum_{S\in\text{\ensuremath{\mathscr{S}}}''\colon v\in N\left(S\right)}}y''\left(S\right)= & {\displaystyle \sum_{u\in ds\colon v\in N\left(\left\{ u\right\} \right)}}y''\left(\left\{ u\right\} \right)\\
= & {\displaystyle \sum_{u\in ds\colon v\in N\left(\left\{ u\right\} \right)}}\frac{99}{500}\cdot\left(w\left(u\right)-2\right)\\
\leq & {\displaystyle \sum_{u\in ds\colon v\in N\left(\left\{ u\right\} \right)}}\frac{99}{500}\cdot\left(\left(w\left(v\right)+2\right)-2\right)\\
= & \left|ds\cap N\left(\left\{ v\right\} \right)\right|\cdot\frac{99}{500}\cdot w\left(v\right)\\
\leq & 5\cdot\frac{99}{500}\cdot w\left(v\right)=\frac{99}{100}\cdot w\left(v\right)=c_{v}\;.
\end{array}
\]

Thus the statement is correct for base case $step=0$. The induction
hypothesis is that the statement is correct up to $step=k$, and we
will prove that the statement remains correct for $step=k+1$. Note,
the only line that can violate the constraints is line \ref{algline_CDS:l15},
where $y''\left(S\right)$ is updated. In step $k+1$, the algorithm
selects node $v$. Therefore, according to Proposition \ref{prop:alg2-packed},
node $v$ is \textit{packed}, i.e. $\mbox{\ensuremath{\sum_{S\in\text{\ensuremath{\mathscr{S}}}''\colon v\in N\left(S\right)}y''\left(S\right)=c_{v}}}$,
and according to Proposition \ref{prop:alg2-remains-packed}, it remains
\textit{packed} until the end of the algorithm. Then, nodes where
the statement can be violated are only those that have not yet been
selected. Let us assume in contradiction that there is a node $u$
that has not yet been \textit{packed} and after step $k+1$ violates
the statement. Let $y''_{k}\left(S\right)$ be the value of $y''\left(S\right)$
at step $k$. So, we can write: 
\[
\begin{array}{rlc}
\sum_{S\in\text{\ensuremath{\mathscr{S}}}''\colon u\in N\left(S\right)}y''\left(S\right)> & c_{u}\\
\sum_{S\in\text{\ensuremath{\mathscr{S}}}''\colon u\in N\left(S\right)}\left(y''_{k}\left(S\right)+g\left(S\right)\epsilon\left(v\right)\right)> & c_{u}\\
\epsilon\left(v\right)> & \frac{c_{u}-\sum_{S\in\text{\ensuremath{\mathscr{S}}}''\colon u\in N\left(S\right)}y''_{k}\left(S\right)}{\sum_{S\in\text{\ensuremath{\mathscr{S}}}''\colon u\in N\left(S\right)}g\left(S\right)}= & \epsilon\left(u\right)\;.
\end{array}
\]
In contradiction to the minimality of $\epsilon\left(v\right)$. Thus
the statement is correct for step $k+1$. Hence, the algorithm satisfies
the constraints.

$\text{\ensuremath{\mathscr{S}}}''$ contains subsets $\mbox{\ensuremath{\left\{ v\right\} \vcentcolon v\in ds}}$
(line \ref{algline_CDS:l1}) and $\mbox{\ensuremath{y''_{S}=0\;\forall S\notin\left\{ \left\{ v\right\} \vcentcolon v\in ds\right\} }}$
(Equation \ref{eq:y''S}), then we get $\mbox{\ensuremath{\sum_{S\in\text{\ensuremath{\mathscr{S}}}\colon v\in N\left(S\right)}y''_{S}=\sum_{S\in\text{\ensuremath{\mathscr{S}}}''\colon v\in N\left(S\right)}y''_{S}}}$.
Since $\mbox{\ensuremath{\sum_{S\in\text{\ensuremath{\mathscr{S}}}''\colon v\in N\left(S\right)}y''\left(S\right)\leq c_{v}}}$,
then $\mbox{\ensuremath{\epsilon\left(v\right)\geq0}}$ (Equation
\ref{eq:epsilon_v}). By combining this with line \ref{algline_CDS:l15},
we get $\mbox{\ensuremath{y''\left(S\right)\geq0}}$. Consequently,
from Equation \ref{eq:y''S}, $\mbox{\ensuremath{y''_{S}\leq y''\left(S\right)}}$.
Recall that $\mbox{\ensuremath{c_{v}\triangleq\frac{99}{100}\cdot w\left(v\right)}}$.
Then the following holds: 
\[
\sum_{S\in\text{\ensuremath{\mathscr{S}}}\colon v\in N\left(S\right)}y''_{S}\leq\sum_{S\in\text{\ensuremath{\mathscr{S}}}''\colon v\in N\left(S\right)}y''\left(S\right)\leq c_{v}=\frac{99}{100}\cdot w\left(v\right)\leq w\left(v\right)\;\forall v\in V\,.
\]
Since also $\mbox{\ensuremath{y''_{S}\geq0}}$ (Equation \ref{eq:y''S}),
then vector $\bar{y}''$ satisfies all the constraints of $\MWCDSwD$.
Hence, $\bar{y}''$ is a feasible dual solution.

Since the diameter of the graph is at least three, so no subset of
a single node is a DS. That is, $\mbox{\ensuremath{f\left(\left\{ v\right\} \right)=1\;\forall v\in ds}}$
(Equation \ref{eq:f}). Combine it with Equation \ref{eq:y''S} and
line \ref{algline_CDS:l21}, we get 
\begin{equation}
\sum_{S\subseteq V\vcentcolon S\neq\emptyset}f\left(S\right)\cdot y''_{S}=\sum_{S\in\text{\ensuremath{\mathscr{S}}}}y''_{S}=\sum_{v\in ds}y''\left(\left\{ v\right\} \right)=LB2\,.\label{eq:LB2}
\end{equation}
Therefore, by Definition \ref{def:feasible_LB}, $LB2$ is a feasible
lower bound and $\bar{y}''$ is its feasible dual solution. Hence,
the lemma is correct.
\end{proof}

\paragraph{Correctness of Algorithm \ref{alg:Tree}.}

From Lemma \ref{lem:Tree}, Algorithm \ref{alg:Tree} returns a data
gathering tree and its corresponding MULE location. From Lemma \ref{lem:=0003B1},
if the graph diameter is at least three then it returns a feasible
approximation ratio.
\begin{lem}
Given a connected UDG with at least two nodes, Algorithm \ref{alg:Tree}
returns a data gathering tree and its corresponding MULE location.
\label{lem:Tree}
\end{lem}
\begin{proof}
According to lines \ref{algline_Tree:l14} and \ref{algline_Tree:l18},
the structure consists of directed edges.\\
Since the node weight function is $0<w\left(\cdot\right)<\infty$
(line \ref{algline_Tree:l3}) and $\left|V\right|\geq2$, then according
to Lemma \ref{lem:IDS} the subset $ds'$ that produces Algorithm
\ref{alg:DS} (line \ref{algline_Tree:l4}) is an IDS, and according
to Lemma \ref{lem:CDS}, the subset $cds'$ produced by Algorithm
\ref{alg:CDS} (line \ref{algline_Tree:l5}) is a CDS. From line \ref{algline_Tree:l7},
also $cds$ is a CDS. Let $SubG$ be the induced sub graph of $cds$
in $G$. So, since $G$ is a connected graph and $cds$ is a CS, then
$SubG$ is a connected graph. By line \ref{algline_Tree:l8}, $r\in ds'$,
therefore, $r$ is a node in $SubG$. In line \ref{algline_Tree:l12},
since the \textsl{BFS algorithm} starts at node $r$ and the graph
is connected, then at line \ref{algline_Tree:l14} at least one incoming
edge must be added to $r$. From lines \ref{algline_Tree:l13}, \ref{algline_Tree:l14},
\ref{algline_Tree:l16} and \ref{algline_Tree:l18}, we see that each
node (except $r$) has a directed outgoing edge in the structure.
Hence, we can say that the structure spans all the nodes in $G$.\\
From the \textsl{BFS} mechanism to finding shortest paths, and since
for each node in $cds\setminus\left\{ r\right\} $, an added edge
is directed to its predecessor, after the \textsl{for} loop of line
\ref{algline_Tree:l13}, $r$ is reachable from all $cds\setminus\left\{ r\right\} $
nodes in the sub graph $SubG$. After the \textsl{for} loop of line
\ref{algline_Tree:l16}, it can be reached from any of the other nodes
of $V$ to some node in $cds$. Therefore, $r$ is also reachable
from all the nodes in the structure.\\
Since $G$ is a connected UDG and $cds$ is a DS, note that each directed
edge in the structure is added only between a pair of neighbors, so
it is also an undirected edge in $G$.\\
We have seen that the structure that constructs Algorithm \ref{alg:Tree}
is a directed spanning tree in $G$, and that $r$ is reachable from
all nodes. Thus, by Definition \ref{def:Data-gathering-tree}, it
can be said that the structure is a data gathering tree and $r$ is
its root.

According to the last iteration where line \ref{algline_Tree:l7}
is performed, we can see that $m$ is the MULE location used for building
$cds$, which is used later on to build the gathering tree. Hence,
the lemma is correct.
\end{proof}
\begin{lem}
Given a connected UDG with a diameter of at least three, Algorithm
\ref{alg:Tree} returns a feasible approximation ratio $\alpha$.\label{lem:=0003B1}
\end{lem}
\begin{proof}
Consider $\bar{y}'$ and $\bar{y}''$, to be the vector solution that
constructs Algorithm \ref{alg:DS} (Equation \ref{eq:y'S}) and Algorithm
\ref{alg:CDS} (Equation \ref{eq:y''S}), respectively. According
to lemmas \ref{lem:LB1} and \ref{lem:LB2}, since node weights are
$\mbox{\ensuremath{w\left(v\right)\triangleq2\cdot\dist mv+\C}}$
and the graph diameter is at least three, then holds $\mbox{\ensuremath{\sum_{S\in\text{\ensuremath{\mathscr{S}}}\colon v\in N\left(S\right)}y'_{S}\leq\frac{1}{100}\cdot w\left(v\right)}}$
and $\mbox{\ensuremath{\sum_{S\in\text{\ensuremath{\mathscr{S}}}\colon v\in N\left(S\right)}y''_{S}\leq\frac{99}{100}\cdot w\left(v\right)}}$
for all $v\in V$. Consider vector $\bar{y}$ (where $y_{S}$ are
its components) which is defined as follows: $\bar{y}\triangleq\bar{y}'+\bar{y}''$.
So we can write: 
\[
\begin{array}{rl}
{\displaystyle \sum_{S\in\text{\ensuremath{\mathscr{S}}}\colon v\in N\left(S\right)}}y_{S}= & {\displaystyle \sum_{S\in\text{\ensuremath{\mathscr{S}}}\colon v\in N\left(S\right)}}\left(\bar{y}'_{S}+\bar{y}''_{S}\right)\\
\leq & \frac{1}{100}\cdot w\left(v\right)+\frac{99}{100}\cdot w\left(v\right)\\
= & w\left(v\right)\;\forall v\in V\,.
\end{array}
\]
Since also $y_{S}=\bar{y}'_{S}+\bar{y}''_{S}\geq0$, then vector $\bar{y}$
satisfies all the constraints of $\MWCDSwD$. Hence, $\bar{y}$ is
a feasible dual solution. Using equations \ref{eq:LB1} and \ref{eq:LB2},
we can also write: 
\[
\sum_{S\subseteq V\vcentcolon S\neq\emptyset}f\left(S\right)\cdot y_{S}=\sum_{S\subseteq V\vcentcolon S\neq\emptyset}f\left(S\right)\cdot y'_{S}+\sum_{S\subseteq V\vcentcolon S\neq\emptyset}f\left(S\right)\cdot y''_{S}=LB1+LB2\,.
\]
Therefore, by Definition \ref{def:feasible_LB}, $LB=LB1+LB2$ is
a feasible lower bound.

In Lemma \ref{lem:Tree}, we showed that $cds$ is a CDS. Then by
line \ref{algline_Tree:l20} and from Definition \ref{def:=0003B1},
we can say that the value $\alpha$ returned by algorithm \ref{alg:Tree}
is a feasible approximation ratio.
\end{proof}

\subsection{Performance guarantees\label{subsec:Performance}}

In this section we show performance guarantees of Algorithm \ref{alg:Tree}
(including algorithms \ref{alg:DS} and \ref{alg:CDS}) in the context
of approximation and runtime. We analyze the approximation ratio of
the gathering tree produced by Algorithm \ref{alg:Tree} with respect
to the optimal gathering tree for the $\MULE$ problem. According
to Theorem \ref{thm:Gathering-tree-approximation-by-MWCDS}, we show
that if the input graph satisfies the condition $\mbox{\ensuremath{\underset{v\in V_{BB}^{\ast}}{Average}\left[\dist{m^{\ast}}v\right]>1.3}}$,
then the algorithm achieves an approximation of $\mbox{\ensuremath{20+\varepsilon}}$.
It should be noted that for graphs that do not meet this condition,
the problem still remains open. From now on, let us consider only
graphs that satisfy the condition. Lemma \ref{lem:diameter>=00003D3}
and Lemma \ref{lem:diameter<3} prove the approximation ratio. Further,
we analyze the algorithm runtime and present optimizations that can
achieve time complexity of $O\left(n^{3}\cdot\varDelta\left(G\right)\right)$,
where $\varDelta\left(G\right)$ is the maximum degree in the graph
and $n$ is the number of nodes.

\paragraph{Algorithm \ref{alg:Tree} produces a $\mbox{\ensuremath{\boldsymbol{\left(20+\varepsilon\right)-approximate}}}$
solution.}

Let us refer to nodes in $cds\setminus ds$ as the ''\textsl{connectors}'',
since they connect the $ds$ nodes.
\begin{lem}
There is an injective mapping function from each connector to a single
node in $ds$, such that no more than two connectors are mapped to
one node. The connector is at most one unit away from the node, and
if there is an additional connector, then it is at most two units
away.\label{lem:map-2-to-1}
\end{lem}
\begin{proof}
In order to prove this claim, we first build a spanning tree $T'=\left(cds,E'\right)$,
which spans all $cds$ nodes. Then, based on $T'$, we explain the
mapping of each \textsl{connector}.

The construction of $T'$ is based on the decisions made by Algorithm
\ref{alg:CDS}. For each node $v$ selected at each step, we add two
types of edges. The first type is for each subset $S\in\text{\ensuremath{\mathscr{S}}}^{2}$,
we add the edge $\left(v,u\right)$ to $E'$, where $\mbox{\ensuremath{u\in S\cap N\left(\left\{ v\right\} \right)\cap ds}}$
if there is one, and if not, $\mbox{\ensuremath{u\in S\cap N\left(\left\{ v\right\} \right)\cap cds}}$.
The second type of edges is for each subset $S\in\text{\ensuremath{\mathscr{S}}}^{1}\setminus\text{\ensuremath{\mathscr{S}}}^{2}$,
we add the edge $\left(v,u\right)$ to $E'$, where $u$ is a neighbor
according to the choice of the algorithm in line \ref{algline_CDS:l10}.
At the end, we go over all edges in $E'$ and delete edges with nodes
in $V\setminus cds$.

The proof that $T'=\left(cds,E'\right)$ is a spanning tree for $cds$
is done by induction on the algorithm steps. Recall that $\text{\ensuremath{\mathscr{S}}}''$
is the collection of all \textsl{active} subsets of Algorithm \ref{alg:CDS}.
We want to show that at each step and for each \textsl{un-restricted}
\textsl{active} subset $S\in\text{\ensuremath{\mathscr{S}}}''$, the
induced subgraph of $S\cap cds$ in $T'$ is a spanning tree for the
nodes in $S\cap cds$. The base case, step $0$, is trivial. The induction
hypothesis is that the statement holds for step $k>0$, and we will
prove that the statement still holds for step $k+1$. Consider $S'$
to be the new subset created in step $k+1$. $S'$ contains all subsets
of $\text{\ensuremath{\mathscr{S}}}^{1}$. From the induction hypothesis
for step $k$ and since each $S''\in\text{\ensuremath{\mathscr{S}}}^{1}$
is \textsl{un-restricted}, the induced subgraph of $S''\cap cds$
in $T'$ is a spanning tree for $S''\cap cds$. Let $v$ be the selected
node. Since $ds$ is a DS, then $v$ has a neighbor in $ds$. At initialization,
this neighbor is contained in some subset. So, according to Proposition
\ref{prop:g(v)>0}, implies that $\left|\text{\ensuremath{\mathscr{S}}}^{1}\right|>0$.
If at the end of the algorithm $v\in cds$, then for each $S''\in\text{\ensuremath{\mathscr{S}}}^{1}$
one edge is added to $E'$, and connects $S''\cap cds$ (also its
sub spanning tree) to $v$. If eventually $v\notin cds$, then $v\notin S''\cap cds$
and its edges are deleted at the end. Since all edges are added in
a star form around $v$, then the result structure is still a tree,
and it also spans node $v$. Hence, it is a spanning tree for $S'\cap cds$.
According to the induction principle, the statement is correct for
every step. In particular for the last step, in which $\text{\ensuremath{\mathscr{S}}}''$
contains only one \textsl{un-restricted} subset $S$. From Proposition
\ref{prop:union-contains-CDS} we know that $cds\subseteq S$, so
$T'$ is a spanning tree for $cds$.

The mapping function between the \textsl{connectors} and the $ds$
nodes is described as follows. Since the algorithm adds to $cds$,
only \textsl{connectors} with at least two edges in $E'$, then the
\textsl{connectors} can not be leaves in $T'$. Consider the root
node $r$, we know that $r\in ds$. We go over all edges of $T'$,
and each edge we directed towards $r$. Each \textsl{connector} is
mapped to the nearest (by hops) node in $ds$ which has a directed
path to that \textsl{connector}. In case of multiple choice, select
one arbitrarily.

The proof that we do not map to one node in $ds$, more than two \textsl{connectors}
that meet the distance requirements is as follows. We will show that
each \textsl{connector} $v\in cds\setminus ds$ has a node $u\in ds$
so that there is a directed path from $u$ to $v$ of a length no
more than two hops. If $v$ has an incoming edge from a node in $ds$,
then we are done. Otherwise, let us consider $v'$ to be one of the
\textsl{connectors} of the incoming edges (note, since $v$ is not
a leaf then there must be one). Since $T'$ is a directed tree towards
a particular node, then $v'$ has only one outgoing edge, and this
edge is to $v$. All other edges of $v'$ must be incoming. We have
already seen that $\left|\text{\ensuremath{\mathscr{S}}}^{1}\right|>0$
since each selected node has a neighbor from $ds$ in $G$, so each
\textsl{connector} has a neighbor from $ds$, also in tree $T'$.
Then, one of the incoming edges of $v'$ must be from a node in $ds$,
and let us denote this node by $u$. Thus, we have the path $u\rightarrow v'\rightarrow v$
in $T'$, and we are done.\\
Since $T'$ is a directed tree towards a particular node, each node
has only one outgoing edge. Therefore, each node from $ds$ can reach
at most two \textsl{connectors} at a distance of two hops in $T'$.
These are the only two \textsl{connectors} that can be mapped to this
node. Since we are dealing with an UDG, the first must be at most
one unit away, and the second at most two units. Hence, we complete
the proof.
\end{proof}
\begin{lem}
Given a connected UDG with \textbf{a diameter of at least three},
Algorithm \ref{alg:Tree} returns a \linebreak{}
$\mbox{\ensuremath{\left(20+\varepsilon\right)-approximate}}$ solution
to the $\MULE$ problem.\label{lem:diameter>=00003D3}
\end{lem}
\begin{proof}
Recall that $w\left(v\right)=2\cdot\dist mv+\C$. According to Lemma
\ref{lem:map-2-to-1}, each node $v$ from $ds$ has at most two \textsl{connectors}
in $cds\setminus ds$, where one \textsl{connector} has a weight greater
by $2$ and the other by $4$, than $v$'s weight. Therefore, we can
bound $cds$ weight by $Weight\left(cds\right)\leq\sum_{v\in ds}\left(3\cdot w\left(v\right)+6\right)$.
In Algorithm \ref{alg:CDS} line \ref{algline_CDS:l2}, we assign
initial $y$'s values, then we can bound $F_{\bar{y}}$ from below
by $F_{\bar{y}}\geq\sum_{v\in ds}\frac{99}{500}\cdot\left(w\left(v\right)-2\right)$.
Since both sums contain only the $ds$ nodes, we can calculate the
ratio even for just a single node. Note that $w\left(v\right)\geq C$,
and since $\mbox{\ensuremath{0<\RM<0.3}}$, then $C>14.5$ (Figure
\ref{fig:C(Rm)}). So from Definition \ref{def:=0003B1}, we can state
that the approximation ratio $\alpha$ is:
\begin{align}
\frac{Weight\left(cds\right)}{F_{\bar{y}}} & \leq\frac{\sum_{v\in ds}\left(3\cdot w\left(v\right)+6\right)}{\sum_{v\in ds}\frac{99}{500}\cdot\left(w\left(v\right)-2\right)}=\frac{3\cdot w\left(v\right)+6}{\nicefrac{99}{500}\cdot\left(w\left(v\right)-2\right)}=\nonumber \\
 & =\frac{3}{\nicefrac{99}{500}}+\frac{6+6}{\nicefrac{99}{500}\cdot\left(w\left(v\right)-2\right)}\leq\frac{500}{99}\left(3+\frac{12}{\left(14.5-2\right)}\right)=20\;.\label{eq:P-G-1}
\end{align}
From Equation \ref{eq:T-R-P-12}, $OPT_{cds^{\ast}}\geq F_{\bar{y}}$,
then we can write:
\begin{equation}
\frac{Weight\left(cds\right)}{OPT_{cds^{\ast}}}\stackrel{\left(\ref{eq:T-R-P-12}\right)}{\leq}\frac{Weight\left(cds\right)}{F_{\bar{y}}}\stackrel{\left(\ref{eq:P-G-1}\right)}{\leq}20\;.\label{eq:P-G-2}
\end{equation}
Let $T$ be the gathering tree that produces Algorithm \ref{alg:Tree}
and $m$ is its MULE location. From Lemma \ref{lem:OPT-bounded-from-above}
and since the tree's BackBone is $cds$, we get:
\begin{equation}
\cost Tm\stackrel{\left(\ref{eq:R-T-M-1}\right)}{\leq}\sum_{v\in V_{BB}}\left(2\cdot\dist mv+\C\right)=\sum_{v\in cds}\left(2\cdot\dist mv+\C\right)=Weight\left(cds\right)\;.\label{eq:P-G-3}
\end{equation}
Let us distinguish between the two MULE locations, $m$ is the one
chosen by the algorithm and $m^{\ast}$ of the optimal solution. Let
$cds_{m}^{\ast}$ be the optimal solution of MWCDS for location $m$,
where $cds_{m}$ is the CDS constructed by the algorithm, and $cds_{m^{\ast}}^{\ast}$
is for location $m^{\ast}$. The algorithm selects the minimal-weight
CDS, so we have the following inequality: 
\begin{equation}
20\cdot OPT_{cds_{m^{\ast}}^{\ast}}\stackrel{\left(\ref{eq:P-G-2}\right)}{\geq}Weight\left(cds_{m^{\ast}}\right)\geq Weight\left(cds_{m}\right)\;.\label{eq:P-G-4}
\end{equation}
Let $T^{\ast}$ be the optimal gathering tree and $m^{\ast}$ is its
MULE location. From Theorem \ref{thm:Gathering-tree-approximation-by-MWCDS},
$OPT_{T^{\ast}}>\left(1+\varepsilon\right)^{-1}\cdot OPT_{cds_{m^{\ast}}^{\ast}}$.
So we can put it all together into one equation and conclude that:
\begin{align*}
\frac{\cost Tm}{OPT_{T^{\ast}}} & \stackrel{\left(\ref{eq:R-T-M-2}\right)}{<}\frac{\cost Tm}{\left(1+\varepsilon\right)^{-1}\cdot OPT_{cds_{m^{\ast}}^{\ast}}}\stackrel{\left(\ref{eq:P-G-3}\right)}{\leq}\left(1+\varepsilon\right)\cdot\frac{Weight\left(cds_{m}\right)}{OPT_{cds_{m^{\ast}}^{\ast}}}\stackrel{\left(\ref{eq:P-G-4}\right)}{\leq}\\
 & \stackrel{\left(\ref{eq:P-G-4}\right)}{\leq}\left(1+\varepsilon\right)\cdot\frac{Weight\left(cds_{m}\right)}{20^{-1}\cdot Weight\left(cds_{m}\right)}=20\cdot\left(1+\varepsilon\right)=20+\varepsilon\;,
\end{align*}
where the last equality is if we include the multiplication by $20$
into the inaccuracy $\varepsilon$. So that \linebreak{}
$\mbox{\ensuremath{\varepsilon\triangleq20\left(\frac{2}{\left(\C+2.6\right)}\cdot\left(\underset{v\in V_{BB}^{\ast}}{Average}\left[\dist{m^{\ast}}v\right]-1.3\right)\right)^{-1}}}$.
\end{proof}
\begin{lem}
Given a connected UDG with \textbf{a diameter of less than three},
Algorithm \ref{alg:Tree} returns a $\mbox{\ensuremath{\left(20+\varepsilon\right)-approximate}}$
gathering tree to the $\MULE$ problem.\label{lem:diameter<3}
\end{lem}
\begin{proof}
Since the graph is an UDG with a bounded area, then for each node
$\mbox{\ensuremath{v\in V}}$ holds $\mbox{\ensuremath{\dist mv\leq2}}$.
Recall that the node weight function is $\mbox{\ensuremath{2\cdot\dist mv+\C}}$,
so we have $\mbox{\ensuremath{Weight\left(cds\right)\leq\left(4+C\right)\cdot\left|cds\right|}}$.
From Lemma \ref{lem:map-2-to-1}, we know that $\mbox{\ensuremath{\left|cds\right|\leq3\cdot\left|ds\right|}}$.
Therefore, we can write the following equation:
\begin{equation}
Weight\left(cds\right)\leq3\cdot\left(4+C\right)\cdot\left|ds\right|\;.\label{eq:P-G-5}
\end{equation}
Note that $\mbox{\ensuremath{w\left(v\right)\geq\C}}$. Let $mcds^{\ast}$
be the optimal solution for the MCDS problem. Thus, since $cds^{\ast}$
is a CDS, it must be hold that $\mbox{\ensuremath{\left|cds^{\ast}\right|\geq\left|mcds^{\ast}\right|}}$.
We obtain the following inequality:
\begin{equation}
OPT_{cds^{\ast}}=\sum_{v\in cds^{\ast}}w\left(v\right)\geq C\cdot\left|cds^{\ast}\right|\geq C\cdot\left|mcds^{\ast}\right|\;.\label{eq:P-G-6}
\end{equation}
Let $mis$ be a maximal independent set. Weili Wu et al. \cite{Wu2006}
showed that for UDG graph, $\mbox{\ensuremath{\left|mis\right|\leq3.8\cdot\left|mcds^{\ast}\right|+1.2}}$.
From Proposition \ref{prop:MIS}, $ds$ is a MIS, then 
\begin{equation}
\left|ds\right|\leq3.8\cdot\left|mcds^{\ast}\right|+1.2\;.\label{eq:P-G-7}
\end{equation}
If we sum it all up here, we get:
\begin{align*}
\frac{Weight\left(cds\right)}{OPT_{cds^{\ast}}} & \stackrel{\left(\ref{eq:P-G-5}\right)}{\leq}\frac{3\cdot\left(4+C\right)\cdot\left|ds\right|}{OPT_{cds^{\ast}}}\stackrel{\left(\ref{eq:P-G-6}\right)}{\leq}\frac{3\cdot\left(4+C\right)\cdot\left|ds\right|}{C\cdot\left|mcds^{\ast}\right|}\stackrel{\left(\ref{eq:P-G-7}\right)}{\leq}\\
 & \stackrel{\left(\ref{eq:P-G-7}\right)}{\leq}\frac{3\cdot\left(4+C\right)\cdot\left(3.8\cdot\left|mcds^{\ast}\right|+1.2\right)}{C\cdot\left|mcds^{\ast}\right|}=\\
 & =\frac{45.6\cdot\left|mcds^{\ast}\right|+11.4\cdot C\cdot\left|mcds^{\ast}\right|+14.4+3.6\cdot C}{C\cdot\left|mcds^{\ast}\right|}\;.
\end{align*}
Recall that $\C>14.5$, and $mcds^{\ast}$ is a DS, so it must contain
at least one node, i.e. $\mbox{\ensuremath{\left|mcds^{\ast}\right|\geq1}}$.
We can continue the inequality as follows:
\begin{equation}
\frac{Weight\left(cds\right)}{OPT_{cds^{\ast}}}<\frac{45.6}{14.5}+11.4+\frac{14.4}{14.5\cdot1}+\frac{3.6}{1}<19.14<20\;.\label{eq:P-G-8}
\end{equation}
As in Lemma \ref{lem:diameter>=00003D3}, let us distinguish between
two MULE locations, $m$ and $m^{\ast}$. Remember that the algorithm
selects the minimal-weight CDS, then we can state that:
\begin{equation}
20\cdot OPT_{cds_{m^{\ast}}^{\ast}}\stackrel{\left(\ref{eq:P-G-8}\right)}{>}Weight\left(cds_{m^{\ast}}\right)\geq Weight\left(cds_{m}\right)\;.\label{eq:P-G-9}
\end{equation}
Similarly to Lemma \ref{lem:diameter>=00003D3}, we can put it all
together and conclude that:
\begin{align*}
\frac{\cost Tm}{OPT_{T^{\ast}}} & \stackrel{\left(\ref{eq:R-T-M-2}\right)}{<}\frac{\cost Tm}{\left(1+\varepsilon\right)^{-1}\cdot OPT_{cds_{m^{\ast}}^{\ast}}}\stackrel{\left(\ref{eq:P-G-3}\right)}{\leq}\left(1+\varepsilon\right)\cdot\frac{Weight\left(cds_{m}\right)}{OPT_{cds_{m^{\ast}}^{\ast}}}\stackrel{\left(\ref{eq:P-G-9}\right)}{<}\\
 & \stackrel{\left(\ref{eq:P-G-9}\right)}{<}\left(1+\varepsilon\right)\cdot\frac{Weight\left(cds_{m}\right)}{20^{-1}\cdot Weight\left(cds_{m}\right)}=20\cdot\left(1+\varepsilon\right)=20+\varepsilon\;,
\end{align*}
where the last equality is if we include the multiplication by $20$
into the inaccuracy $\varepsilon$ (as in Lemma \ref{lem:diameter>=00003D3}).
\end{proof}

\paragraph{The total time complexity can be optimized to $O\left(n^{3}\cdot\varDelta\left(G\right)\right)$.}
\begin{prop}
The number of \textsl{active} subsets maintained by Algorithm \ref{alg:DS}
is at most twice the number of nodes.\label{prop:alg1-num-of-subsets}
\end{prop}
\begin{proof}
Initially, the number of \textsl{active} subsets ($\left|\text{\ensuremath{\mathscr{S}}}'\right|$)
is as the number of nodes in the graph. We will show that for each
node, at most one new subset is added to $\text{\ensuremath{\mathscr{S}}}'$.
Without loss of generality, look at some node $\mbox{\ensuremath{u\in V}}$.
Consider the first iteration in which the \textit{if} statement in
line \ref{algline_DS:l12} is met, for node $u$. Note, $u$ is not
yet selected by the algorithm. The new subset $S'$ is added in line
\ref{algline_DS:l14}, where $u$ is its only neighbor. Since the
assignment $\mbox{\ensuremath{g\left(S'\right)\leftarrow1}}$ is performed,
then only after the algorithm will assign $\mbox{\ensuremath{g\left(S'\right)\leftarrow0}}$,
the statement will be able to met again. The only place where the
assignment $\mbox{\ensuremath{g\left(S'\right)\leftarrow0}}$ can
be performed is in line \ref{algline_DS:l9}, when $u$ is chosen
by the algorithm. In the same step, node $u$ is added to $A$, so
the \textit{if} statement in line \ref{algline_DS:l12} will not be
checked again for $u$. Hence, no other new subset will be added for
node $u$.
\end{proof}
\begin{prop}
At each step of Algorithm \ref{alg:CDS}, a node can be contained
in only one un-restricted \textsl{active} subset. \label{prop:g(v)=00003D1}
\end{prop}
\begin{proof}
Without loss of generality, let us consider node $u\in V$. We will
prove by contradiction, if node $u$ is contained in some \textsl{un-restricted}
subset $\mbox{\ensuremath{S\in\text{\ensuremath{\mathscr{S}}}''}}$
then there is no other \textsl{un-restricted} subset $\mbox{\ensuremath{S'\in\text{\ensuremath{\mathscr{S}}}''}}$
that contains $u$. Let us assume in contradiction that $u$ is contained
in both $S$ and $S'$. Note, $u$ can not be added to the two subsets
in the same step. We assume that $u$ is already contained in subset
$S'$, and let us consider the step where $u$ is added to subset
$S$. According to propositions \ref{prop:alg2-remains-packed} and
\ref{prop:chosen-not-infinity}, since $u$ is contained in $S'$
then it cannot be the selected node. Therefore, $u$ had to be added
to $S$ due to the union in line \ref{algline_CDS:l18} of collection
$\text{\ensuremath{\mathscr{S}}}^{1}$. $S'$ is the only \textsl{un-restricted}
subset containing $u$, so $S'$ must be in the collection. By Proposition
\ref{prop:chosen-not-infinity}, the selected node $v$ can not be
in $S'$. Since $\mbox{\ensuremath{S'\in\text{\ensuremath{\mathscr{S}}}^{1}}}$,
then there must be a neighbor of $v$ is $S'$. Therefore, the \textit{if}
statement in line \ref{algline_DS:l16} is met and the assignment
$\mbox{\ensuremath{g\left(S'\right)\leftarrow0}}$ is performed. But
this is in contradiction to the assumption that $S'$ is \textsl{un-restricted}.
Hence, subset $S$ must be the only \textsl{un-restricted} subset
that contains $u$.
\end{proof}
\begin{prop}
At each step of Algorithm \ref{alg:CDS}, a node can have at most
five adjacent \textsl{un-restricted} subsets. \label{prop:5-adjacent-sets}
\end{prop}
\begin{proof}
We will prove by induction on the steps. Initially, the subsets of
$\text{\ensuremath{\mathscr{S}}}''$ are based on the $ds$ nodes
and they are all \textsl{un-restricted}. Since $ds$ is an IS then
the statement holds for base case, step $0$. The induction hypothesis
is that the statement holds for step $k>0$, and we will prove that
the statement still holds for step $k+1$. We will show that for each
node $\mbox{\ensuremath{u\in V}}$ the condition is preserved. Let
$v$ be the selected node in step $k+1$, and $S$ be the new subset.
If $u$ is not adjacent to any subset of $\text{\ensuremath{\mathscr{S}}}^{1}$
and is not a neighbor of $v$, then the new subset $S$ has no influence
on $u$. If $u$ adjacent to at least one subset of $\text{\ensuremath{\mathscr{S}}}^{1}$,
then after the union in line \ref{algline_CDS:l18} the number of
\textsl{un-restricted} subsets adjacent to $u$ can not increase.
Consider the case where $u$ is not adjacent to any subset of $\text{\ensuremath{\mathscr{S}}}^{1}$
but is neighbor of $v$. Each subset adjacent to $u$ contains a neighbor
of $u$. Note, all of these neighbors are independent, and are also
independent in $v$. Therefore, the number of subsets adjacent to
$u$ can not be greater than five, including the new \textsl{un-restricted}
subset $S$ that contains $v$. Hence, the condition holds for each
step, and the statement is correct.
\end{proof}

\subparagraph{Algorithm \ref{alg:DS} can be implemented in such a way that will
produce a solution in $O\left(n^{2}\cdot\Delta\left(G\right)\right)$
time.}

In order to do so, we maintain for each node $v\in V$ two additional
variables. $sum_{y}\left(v\right)$ will be maintained equivalent
to the value of $\mbox{\ensuremath{\sum_{S\in\text{\ensuremath{\mathscr{S}}}'\colon v\in N\left(S\right)}y\left(S\right)}}$,
and $sum_{g}\left(v\right)$ for the value of $\mbox{\ensuremath{\sum_{S\in\text{\ensuremath{\mathscr{S}}}'\colon v\in N\left(S\right)}g\left(S\right)}}$.
These two types of variables will be used to calculate the potential
function $\epsilon\left(\cdot\right)$ and the \textsl{if} statement
of line \ref{algline_DS:l12}, in constant time. $sum_{y}\left(\cdot\right)$
variables will be initialized to zero, and $sum_{g}\left(\cdot\right)$
variables will be initialized to their node degree since all subsets
are initially \textsl{un-restricted}. The $sum_{y}\left(\cdot\right)$
values should be updated after changing the $y\left(S\right)$ values
in line \ref{algline_DS:l8}. To do so, we will add an update after
this line, for each node $u\in N\left(S\right)$, the $sum_{y}\left(u\right)$
value will be increased by $g\left(S\right)\cdot\epsilon$. The $sum_{g}\left(\cdot\right)$
values should be updated in two places. If the algorithm performs
the assignment $\mbox{\ensuremath{g\left(S\right)\leftarrow0}}$ in
line \ref{algline_DS:l9}, when previously $g\left(S\right)$ value
was $1$, then $sum_{g}\left(u\right)$ should be decrease by $1$
for each $\mbox{\ensuremath{u\in N\left(S\right)}}$. After the assignment
$\mbox{\ensuremath{g\left(S'\right)\leftarrow1}}$ in line \ref{algline_DS:l14},
$sum_{g}\left(u\right)$ should be increased by $1$. Now, to calculate
the algorithm runtime, we will focus within the general \textsl{repeat}
loop, on the \textsl{for} loop in line \ref{algline_DS:l7}. All other
components can be computed in $O\left(n\right)$ time, if we refer
to the complement subsets only symbolically and we consider unsorted
sets. Each operation on $N\left(S\right)$ is performed in $O\left(\Delta\left(G\right)\right)$
time, since a subset can have as many neighbors as a node has. Note
that according to Proposition \ref{prop:alg1-num-of-subsets}, $\mbox{\ensuremath{\left|\text{\ensuremath{\mathscr{S}}}'\right|\leq2n}}$.
The \textsl{for} loop can perform at most $2n$ iterations, as the
number of subsets in $\text{\ensuremath{\mathscr{S}}}'$. By Proposition
\ref{prop:alg1-node-selected-once}, the \textsl{repeat} loop can
perform at most $n$ iterations. Therefore, the total runtime of Algorithm
\ref{alg:DS} is $O\left(n^{2}\cdot\Delta\left(G\right)\right)$.

\subparagraph{Algorithm \ref{alg:CDS} can also be implemented in such a way that
will produce a solution in $O\left(n^{2}\cdot\Delta\left(G\right)\right)$
time.}

Also here (as Algorithm \ref{alg:DS}), we maintain for each node
$v\in V$ the two variables $sum_{y}\left(v\right)$ (for $\mbox{\ensuremath{\sum_{S\in\text{\ensuremath{\mathscr{S}}}''\colon v\in N\left(S\right)}y\left(S\right)}}$)
and $sum_{g}\left(v\right)$ (for $\mbox{\ensuremath{\sum_{S\in\text{\ensuremath{\mathscr{S}}}''\colon v\in N\left(S\right)}g\left(S\right)}}$),
and in addition the variable $\text{\ensuremath{\mathfrak{s}}}\left(v\right)$.
This variable will hold an \textsl{un-restricted} subset that contains
$v$, and from Proposition \ref{prop:g(v)=00003D1} we know that this
subset is the only one. The first two types of variables will be used
to calculate the potential function $\epsilon\left(\cdot\right)$
and the \textsl{if} statement of line \ref{algline_CDS:l7}, in constant
time. While the third will be used in lines \ref{algline_CDS:l5},
\ref{algline_CDS:l6} and \ref{algline_CDS:l16} to find all subsets
adjacent to a node, in $O\left(\Delta\left(G\right)\right)$ time.
Here, variables $sum_{g}\left(v\right)$ will be initialized to $\mbox{\ensuremath{\left|N\left(\left\{ v\right\} \right)\cap ds\right|}}$
for each node $\mbox{\ensuremath{v\in V}}$, $sum_{y}\left(v\right)$
to zeros, and variables $\text{\ensuremath{\mathfrak{s}}}\left(v\right)$
will hold subset $\left\{ v\right\} $ for each $\mbox{\ensuremath{v\in ds}}$
and $\emptyset$ for each $\mbox{\ensuremath{v\in V\setminus ds}}$.
Variables $\text{\ensuremath{\mathfrak{s}}}\left(\cdot\right)$ should
be updated after adding the new subset in line \ref{algline_CDS:l19}.
For each node $u$ in the new \textsl{un-restricted} subset $S$ we
assign $\mbox{\ensuremath{\text{\ensuremath{\mathfrak{s}}}\left(u\right)\leftarrow S}}$.
The $sum_{y}\left(\cdot\right)$ values should be updated after changing
the $y\left(S\right)$ values by the \textsl{for} loop in line \ref{algline_CDS:l14}.
Thus, after line \ref{algline_CDS:l17} we will insert their update.
Let $\text{\ensuremath{\mathscr{S}}}_{u}^{1}$ be the collection of
all \textsl{un-restricted} \textsl{active} subsets adjacent to node
$u$. For each node $\mbox{\ensuremath{u\in V}}$ we find $\text{\ensuremath{\mathscr{S}}}_{u}^{1}$
using $\text{\ensuremath{\mathfrak{s}}}\left(\cdot\right)$, and need
to increase $sum_{y}\left(u\right)$ by $\mbox{\ensuremath{\left|\text{\ensuremath{\mathscr{S}}}_{u}^{1}\right|\cdot\epsilon}}$.
The $sum_{g}\left(\cdot\right)$ values should be recalculated after
changing the $g\left(S\right)$ values in lines \ref{algline_CDS:l16}
and \ref{algline_CDS:l19}. So, for each node $\mbox{\ensuremath{u\in V}}$
we find $\text{\ensuremath{\mathscr{S}}}_{u}^{1}$ using $\text{\ensuremath{\mathfrak{s}}}\left(\cdot\right)$,
and assign $\mbox{\ensuremath{sum_{g}\left(u\right)\leftarrow\left|\text{\ensuremath{\mathscr{S}}}_{u}^{1}\right|}}$.
Let us calculate the algorithm's runtime. The main component of the
runtime is the \textsl{repeat} loop, which according to propositions
\ref{prop:alg2-remains-packed} and \ref{prop:chosen-not-infinity},
performs at most $n$ iterations. According to Proposition \ref{prop:5-adjacent-sets},
$\mbox{\ensuremath{\left|\text{\ensuremath{\mathscr{S}}}^{2}\right|\leq\left|\text{\ensuremath{\mathscr{S}}}^{1}\right|\leq5}}$.
Then, within the \textsl{repeat} loop, no component requires more
than $\mbox{\ensuremath{O\left(n\cdot\Delta\left(G\right)\right)}}$
time to calculate (also the maintenance operations we added for $sum_{y}\left(\cdot\right)$
and $sum_{g}\left(\cdot\right)$). Therefore, the total runtime is
$O\left(n^{2}\cdot\Delta\left(G\right)\right)$.

\subparagraph{The time complexity of Algorithm \ref{alg:Tree} is $O\left(n^{3}\cdot\Delta\left(G\right)\right)$.}

Stems from the fact that the algorithm runs algorithms \ref{alg:DS}
and \ref{alg:CDS}, $n$ times.

\section{Empirical Results\label{sec:Empirical-Results}}

In this section we present simulation results of Algorithm \ref{alg:Tree}.
The simulations were performed on random UDG graphs, in a model of
uniform random distribution of nodes across a given square area. Each
graph is characterized by the size of the square area on which it
is spread, and by the average density of its nodes. The simulation
was performed only on connected graphs. The MULE transmission range
($\RM$) is determined to be $0.2$. Since the goal is to show the
$\alpha$-approximation ratio obtained by the algorithm, the MULE
location for each simulation is predetermined, in order to reduce
runtime. The MULE location is intuitively determined to be the closest
node to the center of the given area. There are two interesting aspects
to explore. One aspect, is the performance of the algorithm for increasing
nodes density, on a given area. For this aspect, we prepared three
sequences of random graphs with the same area. The second aspect is
the performance of the algorithm for increasing area. For this aspect,
we prepared three sequences of random graphs with the same density
of nodes.

The first results we present in Figure \ref{fig:Simulation-results}
(a) are simulations for increasing nodes density. We chose three sizes
of square area of $4$, $16$ and $36$ square units. For each area
and density, we ran five simulations of different random graphs. We
have seen that these areas are sufficient to represent the general
trend. First of all, in the results we can see a convergence trend
of all the simulations to a constant approximation value, as the density
increases. This can be explained by the fact that as the density of
the nodes grows, the selection of nodes for $cds$ is increasingly
based on their geographical location in the area. After all, we are
looking for a set that will dominate the entire area. For a fixed
area, the position will remain the same. Another thing that can be
observed is the convergence of the distribution of the results, as
the area increases. The convergence of the distribution can be explained
as follows. The set of nodes that the algorithm finds, as well as
the optimal solution, must dominate the entire area as the density
increases. So when the area increases, also the weights of the nodes
far from the MULE, are increased respectively. In the end, both sets
($cds$ and the optimal solution) are required to dominate the same
area, so the small shifts to the right and left, less affect the overall
cost. What affects more, is the need to have one dominating node in
the unit disk around each node in the graph, and this is also required
from the optimal solution.

The second results we present in Figure \ref{fig:Simulation-results}
(b) are simulations for increasing area. We chose three different
densities of $2.5$, $10$, and $40$. We have seen that they represent
the general trend. In addition, we illustrated the value of the approximation
inaccuracy $\left(1+\varepsilon\right)$ resulting from the reduction
to the MWCDS problem. Since we can not calculate its exact value because
it depends on the optimal solution, then we used the $cds$ nodes
that the algorithm finds. We denote this estimator by $\hat{\varepsilon}$.
As we explained in the previous simulation results, as the area grows,
then the general form of the optimal solution must be more similar
to the solution found by the algorithm. Since the optimal gathering
tree also needs to span all the nodes in the graph, then the trend
should also be similar, as the area increases. Therefore, the $\hat{\varepsilon}$
value calculated by the average of the $cds$ nodes distances should
behave similarly to $\varepsilon$. In the results, we can see the
convergence of the approximation value to a constant. The explanation
for this is as before, since both the algorithm's solution and the
optimal solution should dominate the entire area, and if a node is
selected slightly to the right or left relative to the optimal solution
node, then the effect becomes negligible relative to the weight of
nodes far from the MULE. Another thing that can be observed is convergence
as the density increases. The explanation for this is, of course,
from the results of the previous simulation, in which we saw this
convergence explicitly. One last thing to note is the downward trend
of the estimator $\left(1+\hat{\varepsilon}\right)$, which gives
us an insight for the value of $\left(1+\varepsilon\right)$. For
the range of area sizes we used, we can already see that the approximation
inaccuracy has a very reasonable value. 

\begin{figure}
\begin{minipage}[t]{0.5\columnwidth}%
\begin{center}
(a) Approximation value for increasing \\
nodes density
\par\end{center}
\begin{center}
\includegraphics[width=1\columnwidth]{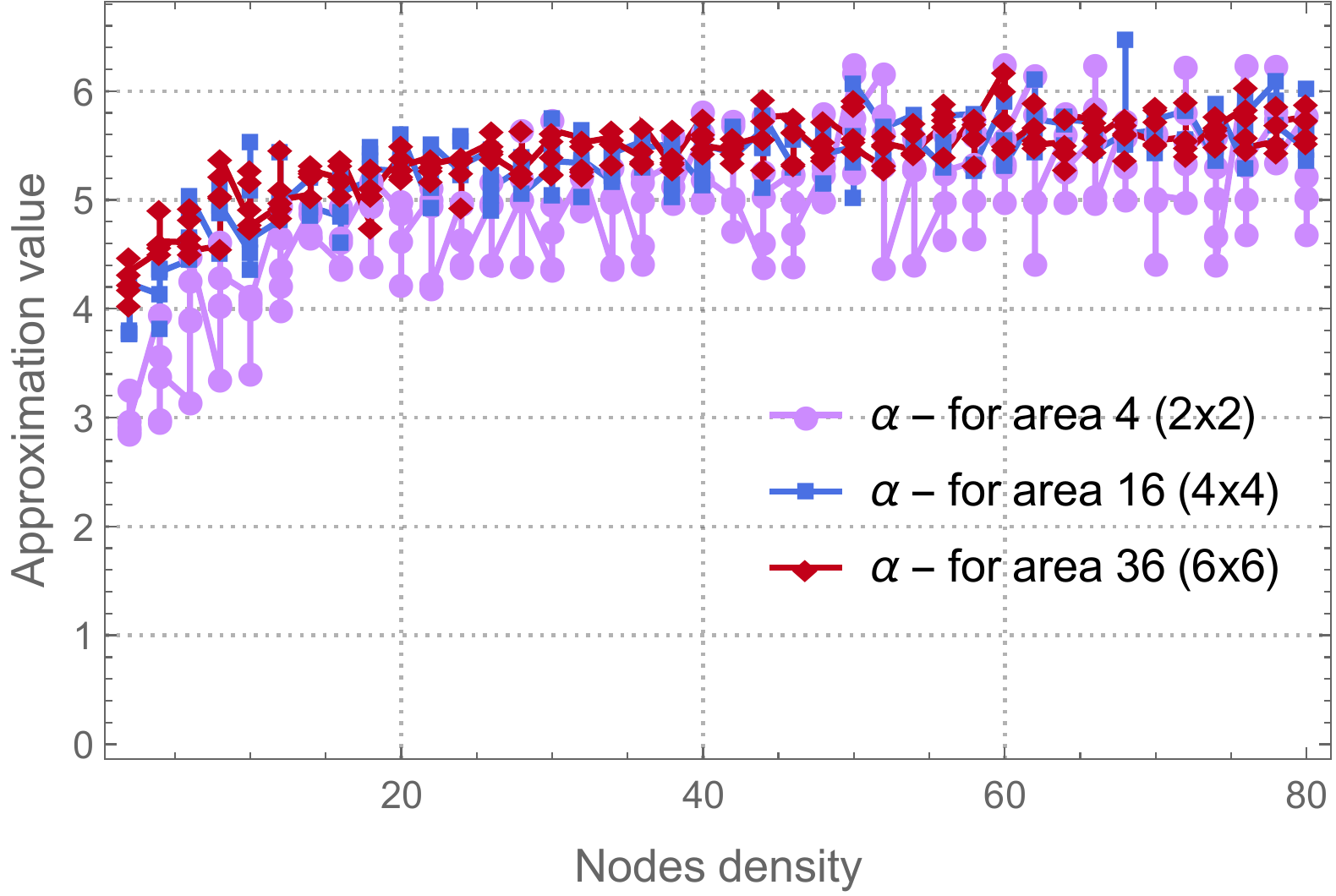}
\par\end{center}%
\end{minipage}%
\begin{minipage}[t]{0.5\columnwidth}%
\begin{center}
(b) Approximation value for increasing \\
area size
\par\end{center}
\begin{center}
\includegraphics[width=1\columnwidth]{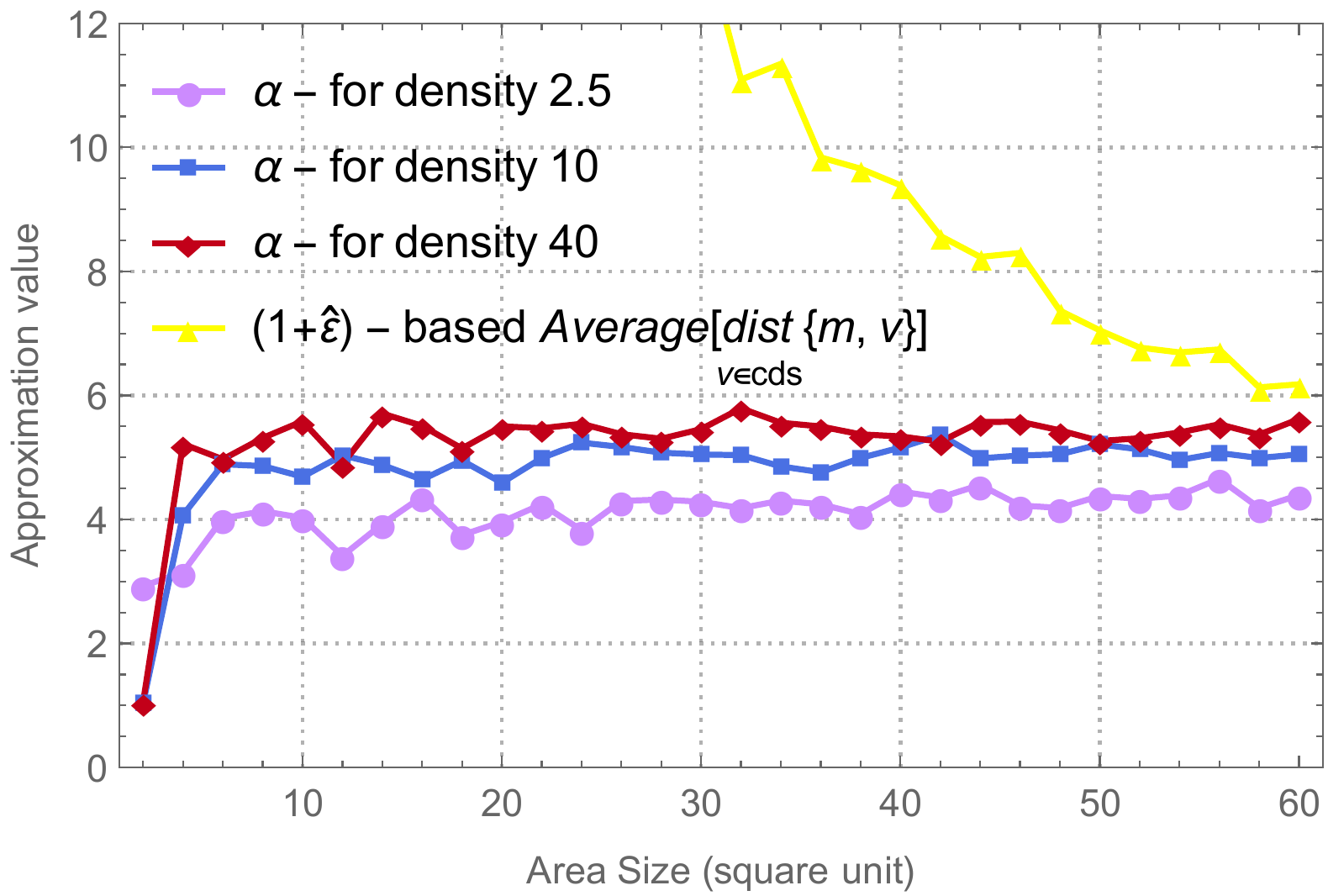}
\par\end{center}%
\end{minipage}

\caption{\label{fig:Simulation-results}Simulation results for the approximation
ratio $\alpha$ of Algorithm \ref{alg:Tree}.}

\end{figure}

\section{Conclusions and Future Work\label{sec:Conclusions}}

In this work, we studied the use of data MULEs in order to gather
data and increase information survivability in wireless sensor networks.
We considered the MULE problem for a general UDG with a single MULE
and only one failed sensor, and referred to this problem as $\MULE$.
We showed that with a reasonable assumption it can be reduced to a
MWCDS problem with an approximation of $\left(1+\varepsilon\right)$.
Then, we proposed a \textsl{primal-dual }algorithm that finds a $20-approximate$
solution to the reduced problem of MWCDS in polynomial time. Finally,
we introduced the complete algorithm which produces a $\left(20+\varepsilon\right)$-approximate
solution for the $\MULE$ problem in $O\left(n^{3}\cdot\varDelta\left(G\right)\right)$
time. Also, we have shown by simulation that in practice the algorithm
achieves even better results. Future extensions of our work could
investigate how to expand the algorithm for multiple MULEs in the
network, or for different cost functions. In addition, it will be
interesting to explore the implications of the algorithm technique
on the research field of the MWCDS problem. 

\begin{flushleft}
\textbf{Acknowledgements:} The research has been supported by Israel
Science Foundation grant No. 317/15 and the US Army Research Office
under grant \#W911NF-18-1-0399.
\par\end{flushleft}


\bibliographystyle{plain}
\phantomsection\addcontentsline{toc}{section}{\refname}\bibliography{Bibliography}

\end{document}